\documentclass{vldb}
\usepackage{amsmath,amssymb,rotating,listings,url}
\usepackage{balance}
\makeatletter
\let\proof\@undefined
\let\endproof\@undefined
\usepackage{amsthm}
\usepackage[ps2pdf, bookmarks, pdfview={FitH}, pdfstartview={FitH}]{hyperref}
\usepackage{balance}

\newcommand{\extended}[1]{#1}
\newcommand{\short}[1]{}

\newcommand{\IEEEtitleabstractindextext}[1]{#1}
\newenvironment{IEEEkeywords}{{\bf Keywords: }}{}

\newlength{\leolen}
\newcommand{\skiptext}[1]{\settowidth{\leolen}{#1}\hspace*{\leolen}}
\newlength{\leoframelength}
\newcommand{\leoframe}[1]{\framebox{\parbox[t]{\leoframelength}{\vspace*{0.5em}#1}}}
\newcommand{\ignore}[1]{}

\newcommand{\s}[1]{\textsf{#1}}
\newcommand{\app}{\mbox{$+\!\!\!+$}}
\newcommand{\gb}{\Uparrow}
\newcommand{\defined}{\triangleq}
\newcommand{\diffuse}[1]{\widehat{\,{#1}\,}}
\newcommand{\diff}[1]{\widetilde{\,{#1}\,}}
\newcommand{\sem}[1]{\mbox{$[\![$}#1\mbox{$]\!]$}}
\newcommand{\trq}[1]{\textrm{$\mathcal{T}_q$}\sem{#1}}
\newcommand{\tre}[1]{\textrm{$\mathcal{T}_e$}\sem{#1}}
\newcommand{\trc}[1]{\textrm{$\mathcal{T}_c$}\sem{#1}}
\newcommand{\monoid}[1]{\textrm{$\mathcal{M}$}\sem{#1}}
\newcommand{\answer}[1]{\textrm{$\mathcal{A}$}\sem{#1}}
\newenvironment{tab}{\vspace{1ex}\begin{minipage}[l]{\textwidth}\sf\begin{tabbing}}{\end{tabbing}\end{minipage}\vspace{0.5ex}}
\newtheorem{theorem}{Theorem}

\newtheorem{lemma}{Lemma}
\newtheorem{algorithm}{Algorithm}
\newtheorem{definition}{Definition}

\extended{
\toappear{}
}

\title{Incremental Query Processing\\ on Big Data Streams}

\author{Leonidas Fegaras\\
University of Texas at Arlington\\
{\tt\small fegaras@cse.uta.edu}}

\begin{document}

\extended{\maketitle}

\IEEEtitleabstractindextext{
\begin{abstract}
\noindent
This paper addresses online query processing for large-scale,
incremental data analysis on a distributed stream processing engine
(DSPE). Our goal is to convert any SQL-like query to an incremental
DSPE program automatically. In contrast to other approaches, we derive
incremental programs that return accurate results, not approximate
answers. This is accomplished by retaining a minimal state during the
query evaluation lifetime and by using incremental evaluation
techniques to return an accurate snapshot answer at each time interval
that depends on the current state and the latest batches of data. Our
methods can handle many forms of queries on nested data collections,
including iterative and nested queries, group-by with aggregation, and
equi-joins. Finally, we report on a prototype implementation of our
framework, called MRQL Streaming, running on top of Spark and we
experimentally validate the effectiveness of our methods.
\end{abstract}

}

\short{\maketitle}

\section{Introduction}

In recent years, large volumes of data are being generated, analyzed,
and used at an unprecedented scale and rate. Data analysis tools that
process these data are typically batch programs that need to work on
the complete datasets, thus repeating the computation on existing data
when new data are added to the datasets. Consequently, batch
processing may be prohibitively expensive for Big Data that change
frequently. For example, the Web is evolving at an enormous rate with
new Web pages, content, and links added daily. Web graph analysis
tools, such as PageRank, which are used extensively by search engines,
need to recompute their Web graph measures very frequently since they
become outdated very fast. There is a recent interest in incremental
Big Data analysis, where data are analyzed in incremental fashion, so
that existing results on current data are reused and merged with the
results of processing the new data. Incremental data processing can
generally achieve better performance and may require less memory than
batch processing for many data analysis tasks.  It can also be used
for analyzing Big Data incrementally, in batches that can fit in
memory.  Consequently, incremental data processing can also be useful
to stream-based applications that need to process continuous streams
of data in real-time with low latency, which is not feasible with
existing batch analysis tools.  For example, the Map-Reduce
framework~\cite{dean:osdi04}, which was designed for batch processing,
is ill-suited for certain Big Data workloads, such as real-time
analytics, continuous queries, and iterative algorithms. New
alternative frameworks have emerged that address the inherent
limitations of the Map-Reduce model and perform better for a wider
spectrum of workloads. Currently, among them, the most promising
frameworks that seem to be good alternatives to Map-Reduce while
addressing its drawbacks are Google's Pregel~\cite{pregel:podc09},
Apache Spark~\cite{spark}, and Apache Flink~\cite{flink}, which are
in-memory distributed computing systems. There are also quite a few
emerging distributed stream processing engines (DSPEs) that realize
online, low-latency data processing with a series of batch
computations at small time intervals, using a continuous streaming
system that processes data as they arrive and emits continuous
results. To cope with blocking operations and unbounded memory
requirements, some of these systems build on the well-established
research on data streaming based on sliding windows and incremental
operators~\cite{babcock:pods02}, which includes systems such as
Aurora~\cite{aurora} and Telegraph~\cite{telegraphCQ}, often yielding
approximate answers, rather than accurate results.  Currently, among
these DSPEs, the most popular platforms are Twitter's (now Apache)
Storm~\cite{peng:osdi10}, Spark's D-Streams~\cite{spark:sosp13}, Flink
Streaming~\cite{flink}, Apache S4~\cite{S4}, and Apache
Samza~\cite{samza}.

This paper addresses online processing for large-scale, incremental
computations on a distributed processing platform. Our goal is to
convert any batch data analysis program to an incremental distributed
stream processing (DSP) program automatically, without requiring the
user to modify this program. We are interested in deriving DSP
programs that produce accurate results, rather than approximate
answers. To accomplish this task, we would need to carefully analyze
the program to identify those parts that can be used to process the
incremental batches of data and those parts that can be used to merge
the current results with the new results of processing the incremental
batches. Such analysis is hard to attain for programs written in an
algorithmic programming language but can become more tractable if it
is performed on declarative queries. Fortunately, most programmers
already prefer to use a higher-level query language, such as Apache
Hive~\cite{hive}, to code their distributed data analysis
applications, instead of coding them directly in an algorithmic
language, such as Java. For instance, Hive is used for over 90\% of
Facebook Map-Reduce jobs. There are many reasons why programmers
prefer query languages. First, it is hard to develop, optimize, and
maintain non-trivial applications coded in a general-purpose
programming language. Second, given the multitude of the new emerging
distributed processing frameworks, such as Spark and Flink, it is hard
to tell which one of them will prevail in the near future. Data
intensive applications that have been coded in one of these paradigms
may have to be rewritten as technologies evolve. Hence, it would be
highly desirable to express these applications in a declarative query
language that is independent of the underlying distributed
platform. Furthermore, the execution of such queries can benefit from
cost-based query optimization and automatic parallelism, thus
relieving the application developers from the intricacies of Big Data
analytics and distributed computing. Therefore, our goal is to convert
batch SQL-like queries to incremental DSP programs. We have developed our
framework on Apache MRQL~\cite{mrql}, because it is both
platform-independent and powerful enough to express complex data
analysis tasks, such as PageRank, data clustering, and matrix
factorization, using SQL-like syntax exclusively. Since we are
interested in deriving incremental programs that return accurate
results, not approximate answers, our focus is in retaining a minimal
state during the query evaluation lifetime and across iterations, and
use incremental evaluation techniques to return an accurate snapshot
answer at each time interval that depends on the current state and the
latest batches of data.

We have developed general, sound methods to transform batch queries to
incremental queries.  The first step in our approach is to transform a
query so that it propagates the join and group-by keys to the query
output.  This is known as lineage
tracking~(\cite{benjelloun:vldb06,bhagwat:vldb04,cui:vldb01}).  That way, the
values in the query output are grouped by a key combination, which
corresponds the join/group-by keys used in deriving these values
during query evaluation.  If we also group the new data in the same
way, then computations on current data can be combined with the
computations on the new data by joining the data on these keys.  This
approach requires that we can combine computations on data that have the
same lineage to derive incremental results.  In our framework, this is
accomplished by transforming a query to a 'monoid homomorphism' by
extracting the non-homomorphic parts of the query outwards, using
algebraic transformation rules, and combine them to form an answer
function, which is detached from the rest of the query.
We have implemented our incremental processing framework using Apache
MRQL~\cite{mrql} on top of Apache Spark Streaming~\cite{spark:sosp13},
which is an in-memory distributed stream processing platform. 
Our system is called {\em Incremental MRQL}. MRQL is
currently the best choice for implementing our framework because other
query languages for data-intensive, distributed computations provide
limited syntax for operating on data collections, in the form of
simple relational joins and group-bys, and cannot express complex data
analysis tasks, such as PageRank, data clustering, and matrix
factorization, using SQL-like syntax exclusively.  Our framework
though can be easily adapted to apply to other query languages, such
as SQL, XQuery, Jaql, and Hive. 

The contribution of this work can be summarized as follows:
\begin{itemize}
\item We present a general automated method to convert most distributed
  data analysis queries to incremental stream processing programs.
\item Our methods can handle many forms of queries, including
  iterative and nested queries, group-by with aggregation, and joins
  on one-to-many relationships.
\extended{\item We have extended our framework to handle deletions
by using state increments to diminish the state in such a way that
the query results would be the same as those we could have gotten on
current data after deletion.}
\item We report on a prototype implementation of our framework using
  Apache MRQL running on top of Apache Spark Streaming. We show the
  effectiveness of our method through experiments on four
  queries: groupBy, join-groupBy, k-means clustering, and PageRank.
\end{itemize}

The rest of this paper is organized as follows.
Section~\ref{approach} illustrates our approach
through examples.  Section~\ref{related-work} compares our work with
related work.  Section~\ref{mrql-language} describes our earlier work
on MRQL query processing. Section~\ref{mrql-algebra} defines the
algebraic operators used in the MRQL algebra.
Section~\ref{normalization} presents transformation rules for
converting algebraic terms to a normal form that is easier to analyze.
Section~\ref{monoid-inference} describes an inference algorithm that
statically infers the merge function that merges the current results
with the results of processing the new data.  Section~\ref{injection}
describes our method for transforming algebraic terms to propagate all
keys used in joins and group-bys to the query
output. Section~\ref{factoring} describes an algorithm that pulls the
non-homomorphic parts of a query outwards, deriving a
homomorphism. Section~\ref{example}, gives an example of the
transformation of a batch query to an incremental query.
\extended{Section~\ref{deletion} extends our framework to handle data deletion.}
Section~\ref{implementation} gives some implementation details.
Section~\ref{performance} presents experiments that evaluate the
performance of our incremental query processing techniques for four
queries.

\section{Highlights of our Approach}\label{approach}

A dataset in our framework is a bag (multiset) that consists of
arbitrarily complex values, which may contain nested bags and
hierarchical data, such as XML and JSON fragments.  We are considering
continuous queries over a number of streaming data sources, $S_i$, for
$0< i\leq n$. A data stream $S_i$ in our framework consists of an
initial dataset, followed by a continuous stream of incremental
batches $\Delta S_i$ that arrive at regular time intervals $\Delta
t$. In addition to streaming data, there may be other input data
sources that remain invariant through time. A streaming query in our
framework can be expressed as $q(\overline{S})$, where an
$S_i\in\overline{S}$ is a streaming data source.  Incremental stream
processing is feasible when we can derive the query results at time
$t+\Delta t$ by simply combining the query results at time $t$ (i.e.,
the current results) with the results of processing the incremental
batches $\Delta S_i$ only, rather than the entire streams
$S_i\uplus\Delta S_i$, where $\uplus$ is the additive (bag)
union. This is possible if $q(\overline{S\uplus\Delta S})$ can be
expressed in terms of $q(\overline{S})$ (the current query result)
and $q(\overline{\Delta S})$ (the incremental query result), that is,
when $q(\overline{S})$ is a homomorphism\footnote{We use the term {\em
homomorphism} throughout the paper as an abbreviation of {\em monoid
homomorphism}.}  over $\overline{S}$. But some queries, such as
counting the number of distinct elements in a stream or calculating
average values after a group-by, are not homomorphisms. For such
queries, we break $q$ into two functions $a$ and $h$, so that
$q(\overline{S})=a(h(\overline{S}))$ and $h$ is a homomorphism. The
function $h$ is a homomorphism if $h(\overline{S\uplus\Delta S})
=h(\overline{S})\otimes h(\overline{\Delta S})$ for some monoid
$\otimes$ (an associative function with a zero element
$\otimes_z$). For example, the query that counts the number of
distinct elements can be broken into the query $h$ that returns the
list of distinct elements (a homomorphism), followed by the answer
query $a$ that counts these elements. Ideally, we would like most of
the computation in $q$ to be done in $h$, leaving only some
computationally inexpensive data mappings to the answer function
$a$. Note that, the obvious solution when $a$ is equal to $q$ and $h$
is the union of data sources, is also the worst-case scenario that we
try to avoid, since it basically requires to compute the new result
from the entire input, $\overline{S\uplus\Delta S}$. On the other
hand, in the special case when $a$ is the identity function and
$\otimes$ is equal to $\uplus$, the output at each time interval can
be simply taken to be only $h(\overline{\Delta S})$, which is the
output we would expect to get from a fixed window system (i.e., new
batches of output from new batches of input).

If we split $q$ into a homomorphism $h$ and an answer function $a$,
then we can calculate $h$ incrementally by storing its results into a
state and then using the current state to calculate the next $h$ result.
Initially, $\mathrm{state}=\otimes_z$ or, if there are initial stream
data, $\mathrm{state}=h(\overline{S})$. Then, at each
time interval $\Delta t$, the query answer is calculated from the state,
which becomes equal to $h(\overline{S\uplus\Delta S})$:
\[\begin{array}{l}
\mathrm{state}\;\leftarrow\;\mathrm{state}\,\otimes\, h(\overline{\Delta S})\\
\mathbf{return}\; a(\mathrm{state})
\end{array}\]
In Spark, for example, the state and the invariant data sources are
stored in memory as Distributed DataSets (Spark's
RDDs~\cite{{spark:nsdi12}}) and are distributed across the worker
nodes. However, the streaming data sources are implemented as
Discretized Streams (Spark's D-Streams~\cite{spark:sosp13}),
which are also distributed.

Our framework works best for queries whose output is considerably smaller
than their input, such as for data analysis queries that
aggregate data. Such queries would require a smaller state and impose
less processing overhead to $\otimes$.

Figure~1 shows the evaluation of an incremental query
$q(S_1,S_2)=a(h(S_1,S_2))$ over two streaming
data sources, where
$h(S_1\uplus\Delta S_1,S_2\uplus\Delta S_2)=h(S_1,S_2)\otimes h(\Delta S_1,\Delta S_2)$.

\begin{figure}
\begin{center}
\scalebox{0.15}{\includegraphics{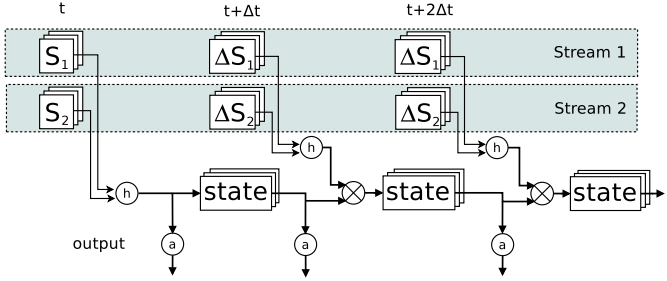}}
\end{center}
\caption{Incremental Query Processing}\label{stream.fig}
\end{figure}

Our query processing system performs the following tasks:
\begin{enumerate}
\item It pulls all non-homomorphic parts of a query $q$ out from the query,
using algebraic transformations.
\item It collects these non-homomorphic parts into an answer function $a$,
leaving an algebraic homomorphic term $h$, such that $q(\overline{S})=a(h(\overline{S}))$.
\item From the homomorphic algebraic term $h(\overline{S})$, our system derives a merge function
$\otimes$, such that $h(\overline{S\uplus\Delta S})=h(\overline{S})\otimes
h(\overline{\Delta S})$.
\end{enumerate}
These tasks are, in general, hard to attain for a program
expressed in an algorithmic programming language. Instead, these tasks
can become more tractable if they are performed on higher-order
operations, such as the MRQL query
algebra~\cite{edbt12,webdb11}. Although all algebraic operations used
in MRQL are homomorphic, their composition may not be. We have
developed transformation rules to derive homomorphisms from
compositions of homomorphisms, and for pulling
non-homomorphic parts outside a query. Our methods can handle
most forms of queries on nested data sets, including iterative
queries, complex nested queries with any form and any number of
nesting levels, general group-bys with aggregations, and general
one-to-one and one-to-many equi-joins.
Our methods cannot handle non-equi-joins and many-to-many equi-joins,
as they are notoriously difficult to implement efficiently in a
streaming or an incremental computing environment.

\paragraph*{\bf Example}
For example, consider the following query (all queries are expressed in MRQL):
\begin{lstlisting}
#$q(S_1,S_2)$# = select (x.A, avg(y.D))
             from x in #$S_1$#, y in #$S_2$#
            where x.B = y.C
            group by x.A
\end{lstlisting}
where $S_1$ and $S_2$ in $q(S_1,S_2)$ are streaming data sources.
Unfortunately, $q(S_1,S_2)$ is not a homomorphism over $S_1$ and $S_2$,
that is, $q(S_1\uplus\Delta S_1,S_2\uplus\Delta S_2)$ cannot be expressed
in terms of $q(S_1,S_2)$ and $q(\Delta S_1,\Delta S_2)$ exclusively.
The intrinsic reason behind this is that there is no lineage in the
query output that links a pair in the query result to the join key (\s{x.B} or \s{y.C})
that contributed to this pair. Consequently, there is no way to tell
how the new data batches $\Delta S_1$ and $\Delta S_2$ will contribute to
the previous query results if we do not know how these results are
related to the inputs $S_1$ and $S_2$.
To compensate, we need to establish links between the query results and the
data sources that were used to form their values. This is called {\em
lineage tracking} and has been used for consistent representation
of uncertain data~\cite{benjelloun:vldb06} and for propagating
annotations in relational queries~\cite{bhagwat:vldb04}. In our case, this lineage
tracking can be accomplished by propagating all keys used in
joins and group-bys along with the values associated with the keys, so that,
for each combination of keys, we have one group of result values. For our query,
this is done by including the join key \s{x.B} as a group-by key.
That is, in the following transformed query:
\begin{lstlisting}
#$h(S_1,S_2)$# = select ((x.A, x.B), (sum(y.D), count(y.D)))
             from x in #$S_1$#, y in #$S_2$#
            where x.B = y.C
            group by x.A, x.B
\end{lstlisting}
the join key is propagated to the output values so that the avg components, sum and count,
are aggregations over groups that correspond to unique combinations of \s{x.A} and \s{x.B}.
It can be shown that this query is a homomorphism over $S_1$ and $S_2$, provided that the
join is not on a many-to-many relationship.
In general, a query with $N$ joins/group-bys/order-bys will be transformed
to a query that injects the join/group-by/order-by keys to the output
so that each output value is annotated with a combination of $N$ keys.
Note that the output of the query $h$ is larger than that of the original query $q$
because it creates more groups and each group is assigned two values (sum and count),
instead of one (avg). This is expected since we extended $q$ with lineage tracking.
The answer query $a$ that gives the final result in $q(S_1,S_2)=a(h(S_1,S_2))$ is:
\begin{lstlisting}
#$a(X)$# = select (k, sum(s)/sum(c))
         from ((k,j),(s,c)) in #$X$#
        group by k
\end{lstlisting}
that is, it removes the lineage $j$ from $X$ but also groups the
result by the group-by key again and calculates the final avg
values. The merge function for the homomorphism $h$ is a full outer join
on the lineage key $\theta$ that aggregates the matches. It is specified as:
\begin{lstlisting}
#$X\otimes Y$# = select (#$\theta$#, (sx+sy,cx+cy))
           from (#$\theta$#,(sx,cx)) in #$X$#,
                (#$\theta$#,(sy,cy)) in #$Y$#
          union select y from y in #$Y$#
                where #$\pi_1$#(y) not in #$\pi_1(X)$#
          union select x from x in #$X$#
                where #$\pi_1$#(x) not in #$\pi_1(Y)$#
\end{lstlisting}
where $\theta$ matches the lineage $(k,j)$
and $\pi_1$ is pair and bag projection. 

\paragraph*{\bf Merging States}
The effectiveness of our incremental processing depends on the efficient
implementation of the state transformation that merges the previous
state with the results of processing the new data,
$\mathrm{state}\otimes h(\overline{\Delta S})$.
The merge operation $X\otimes Y$ in the previous example can be
implemented efficiently as a partitioned join.
On Spark, for example, both the state and the new results are kept in
the distributed memory of the worker nodes (as RDDs), while the full outer
join can be implemented as a coGroup operation, which shuffles the
join input data across the worker nodes using hash
partitioning. However, when the new state is created by coGroup, it is
already partitioned by the join key and is ready to be used for the
next call to coGroup to handle the next batch of data. Consequently, only the
results of processing the new data, which are typically smaller than
the state, would have to be shuffled across the worker nodes before
coGroup. Although other queries may require different merge functions,
the correlation between the previous state and the results of
processing the new data is always based on the lineage keys.
Therefore, regardless of the query, we can keep the state partitioned
on the lineage keys by simply leaving the new state partitions at the
place they were generated. However, the new results would have to be
partitioned and shuffled across the working nodes to be combined with
the current state.

Our approach is based on the assumption that, since the state is kept
in the distributed memory, there will be very little overhead in
replacing the current state with a new state, as long as it is not
repartitioned. But this assumption may not be valid if the input data
or the current state is larger than the available distributed memory.
Indeed, one of our goals is to be able to process data larger than the
available memory, by processing these data incrementally, in batches
that can fit in memory. However, if the state is stored on disk,
replacing it with a new state may become prohibitively
expensive. Fortunately, there is no need for using a distributed
key-value store, such as HBase, to update the state
in-place. There is, in fact, a simpler solution: since the state is
kept partitioned on the lineage keys, every worker node may store its
assigned state partition on its local disk, as a binary file (such as,
an HDFS Sequence file) sorted by the lineage key. Then, after the
results of processing the new data are distributed to worker nodes
(using uniform hashing based on lineage key), each worker node will
sort its new data partition by the lineage key and will merge it with
its old state (stored on its local disk) to create a new state.\ignore{This
state merging can be done fast if the state data are
stored in consecutive blocks on disk.}

\paragraph*{\bf Iteration}
Given that many important data analysis and mining algorithms,
such as PageRank and k-means clustering, require repetition, we have
extended our methods to include repetition, so that these algorithms
too can become incremental. A repetition can take the following
general form:
\[\begin{array}{l}
X\;\leftarrow\;g(\overline{S})\\
\mathbf{for}\;\;i\leftarrow 1\ldots n\\
\skiptext{for(}\,X\;\leftarrow\;f(X,\overline{S})
\end{array}\]
where $X$ is the fixpoint of the repetition that has initial value $g(\overline{S})$.
That is, $X=f^n(g(\overline{S}),\overline{S})$,
where $f^n$ applies $f$ $n$ times.
Note that $X$ is not necessarily a bag. For example, non-negative Matrix
Factorization\extended{~\cite{systemML}, used in machine learning applications,
such as for recommender systems,} splits a matrix $S$ into two
non-negative matrices $W$ and $H$. In that case, the fixpoint $X$
is the pair $(W,H)$, which is refined at every loop step.

An exact incremental solution of the above repetition is only possible
if $f$ is a homomorphism; a strict requirement that excludes many
important iterative algorithms, such as PageRank and k-means
clustering. Instead, our approach is to use an approximate solution
that works well for iterative queries that improve a solution at each
iteration step. As before, we split $f$ into a homomorphism $h$, for
some monoid $\otimes$, and an answer function $a$, so that
$f(X,\overline{S})=a(h(X,\overline{S}))$ and
$h(X,\overline{S\uplus\Delta S}) =h(X,\overline{S})\otimes
h(X,\overline{\Delta S})$.  Let $T$ be the current state on input
$\overline{S}$ and $T'$ be the new state on input
$\overline{S\uplus\Delta S}$.  An ideal solution would have been to
derive the new state $T'$ incrementally, as $T\otimes\Delta T$, so
that the state increment $\Delta T$ depends on $\Delta S$ but not on
$T$.  But this independence from $T$ is not always possible for many
iterative algorithms.  In PageRank, for example, we cannot just
calculate the PageRanks of the new data and merge them with the
PageRank of the existing data because the new PageRank contributions
may have to propagate to the rest of the graph.  On the other hand, it
would be too expensive to correlate the entire state $T$ with $\Delta
T$ at each iteration step. Our compromise is to partially correlate
$T$ with $\Delta T$ at each iteration step, and then fully merge $\Delta T$ with
$T$ after the iteration.  This partial correlation is accomplished
with the operation $\diffuse{\otimes}$, called {\em diffusion}, which
is directly derived from the merge function $\otimes$.  That is, while
$T\otimes\Delta T$ returns a new complete state by merging $T$ with
$\Delta T$, the operation $T\diffuse{\otimes}\Delta T$ returns a new
$\Delta T$ that contains only the part of $T\otimes\Delta T$ that is
different from $T$. Although $T\diffuse{\otimes}\Delta T$ is larger
than $\Delta T$, it is often far smaller than $T\otimes\Delta T$.  For
instance, if $T$ were a bag, $T\diffuse{\otimes}\Delta T$ would have
been a subset of $T\otimes\Delta T$.  Based on this analysis, we are
using the following approximate algorithm to compute the iteration
incrementally:
\[\begin{array}{l}
\Delta X\;\leftarrow\; g(\overline{\Delta S})\\
\mathbf{for}\;\;i\leftarrow 1\ldots n\\
\skiptext{for(}\Delta T\;\leftarrow\;T\,\diffuse{\otimes}\,h(\Delta X,\overline{\Delta S})\\
\skiptext{for(}\Delta X\;\leftarrow\;a(\Delta T)\\
T\;\leftarrow\;T\otimes\Delta T\\
\mathbf{return}\; a(T)
\end{array}\]
The diffusion operator $\diffuse{\otimes}$ must satisfy the property:
\[T\otimes(T\diffuse{\otimes}\Delta T)=T\otimes(T\otimes\Delta T)\]
that is, when $T\diffuse{\otimes}\Delta T$ and $T\otimes\Delta T$
are merged with $T$, they give the same answer, because $T\diffuse{\otimes}\Delta T$
discards the parts of $T$ that are not joined with $\Delta T$,
but these parts are embedded back in the answer when we merge with $T$.

\ignore{
\[T\otimes(T\diffuse{\otimes}h(a(\Delta T),\overline{\Delta S}))\;=\;
h(a(T\otimes\Delta T),\overline{S\uplus\Delta S})\]
To prove that this algorithm calculates the correct $X$, we prove by
induction that $X_i=a(T\otimes\Delta T_i)$, where the subscript $i$
refers to the value at the $i$th step: $a(T\otimes\Delta
T_i)=a(T\otimes(T\diffuse{\otimes}h(\Delta X_{i-1},\overline{\Delta
S}))) =a(h(a(T\otimes\Delta T_{i-1}),\overline{S\uplus\Delta
S}))=a(h(X_{i-1},\overline{S\uplus\Delta S}))=X_i$.
}

For example, PageRank is an iterative algorithm that calculates the
PageRank of a graph node as the sum of the incoming
PageRank contributions from its neighbors, while its own PageRank is equally
distributed to its outgoing neighbors. The PageRank query expressed in
MRQL is a simple self-join on the graph, which is optimized into a
single group-by operation~\cite{edbt12}.  For PageRank, the state
merging $\otimes$ is a full outer join that incorporates the new
PageRank contributions to the existing PageRanks.  The diffusion
operation $T\,\diffuse{\otimes}\,\Delta T$, on the other hand,
propagates the new PageRank contributions from $\Delta T$ to $T$, that
is, from the nodes in $\Delta T$ to their immediate outgoing neighbors
in both $T$ and $\Delta T$.  Thus, at each iteration step, $\Delta T$
is expanded to $T\,\diffuse{\otimes}\,\Delta T$, growing one level at
a time, in a way similar to breadth-first-search. Consequently, our
approximate algorithm propagates PageRanks up to depth $n$, starting
from the $\Delta T$ nodes, and only the affected nodes will be part of
the new $\Delta T$. The $\diffuse{\otimes}$ operation is similar to
$\otimes$, but with a right-outer join instead of a full outer join.
That way, the data from $T$ that are not joined with $\Delta T$ will
not appear in the new $\Delta T$.

A more complete example is the following query that implements the k-means
clustering algorithm by repeatedly deriving $k$ new centroids from the old:
\begin{lstlisting}
repeat centroids = ...
  step select < X: avg(s.X), Y: avg(s.Y) >
         from s in #$S_1$#
        group by k: (select c from c in centroids
                      order by distance(c,s))[0]
\end{lstlisting}
where $S_1$ is the input stream of points on the X-Y plane,
\s{centroids} is the current set of centroids ($k$ cluster centers),
and \s{distance} is a function that calculates the distance between
two points. The initial value of \s{centroids} (the \s{...} value) can
be a bag of $k$ random points. The inner select-query in the group-by
part assigns the closest centroid to a point \s{s} (where \s{[0]}
returns the first tuple of an ordered list). The outer select-query in
the repeat step clusters the data points by their closest centroid,
and, for each cluster, a new centroid is calculated from the average
values of its points.
As in the previous join-groupBy example, the average value of a bag of values is decomposed into a pair 
that contains the sum and the count of values.
That is, the state is a bag of $\{(k, ((sx,cx),(sy,cy)))\}$
so that the centroids are the bag of points $\{(sx/cx,sy/cy)\}$
and $k$ is the lineage (the group-by key), which is the X-Y coordinates of a centroid.
Consequently, the answer query that returns the final result (the centroids) is:
\begin{lstlisting}
#$a$#(state) = select < X: sx/cx, Y: sy/cy >
             from (k,((sx,cx),(sy,cy))) in state
\end{lstlisting}
(it does not require a group-by since $k$ is the only lineage key),
while $h(X,S_1)$ is:
\begin{lstlisting}
select (k, ( (sum(s.X),count(s.X)),
             (sum(s.Y),count(s.Y)) ))
  from s in #$S_1$#
 group by k: (select c from c in #$X$#
               order by distance(c,s))[0]
\end{lstlisting}
The merge function $\otimes$ is a full outer join, similar to
the one used by the join-groupBy example.  The diffusion operation
$X\diffuse{\otimes}Y$ though is a right-outer join that discards those
state data that do not join with the new data. That is, it is equal to
the $X\otimes Y$ query without the last union:
\begin{lstlisting}
#$X\,\diffuse{\otimes}\, Y$# = select (k, (sx+sy,cx+cy))
           from (k,(sx,cx)) in #$X$#,
                 (k,(sy,cy)) in #$Y$#
          union select y from y in #$Y$#
                 where #$\pi_1$#(y) not in #$\pi_1(X)$#
\end{lstlisting}

\extended{
But, suppose that one of the group-by or join keys in a query is a
floating point number. This is the case with the previous k-means
clustering query, because it groups the points by their closest
centroids, which contain floating point numbers. The group-by
operation itself is not a problem because the centroids are fixed
during the group-by. The problem arises when merging the current state
with the results of processing the new data. For the k-means query
example, the merge function $\otimes$ is an equi-join whose join
attribute is a centroid, so that the sum and count values associated
with the same centroid are brought together from the join inputs and
are accumulated. Since the join condition is over attributes with
floating point numbers, the join condition will fail in most
cases. This problem becomes even worse for the approximate solution,
because it uses different sets of centroids $X$ and $X_{prev}$ when
states are merged. Most iterative queries do not have this
problem. The lineage in PageRank, for example, is the node ID, which
remains invariant across iterations. For these uncommon queries, such
as k-means, that have floating point numbers in their
join/group-by/order-by attributes, we use yet another approximation:
the join is done based on an ``approximate equality" where two
floating point numbers are taken to be equal if their difference is
below some given threshold. This works well for our approximate
solution for iteration because it is based on the assumption that the
new solution $X$ is approximately equal to the previous one,
$X_{prev}$.
}

\section{Related Work}\label{related-work}

New frameworks in distributed Big Data analytics have become essential
tools to large-scale machine learning and scientific discoveries.
Among these frameworks, the Map-Reduce programming
model~\cite{dean:osdi04} has emerged as a generic, scalable, and cost
effective solution for Big Data processing on clusters of commodity
hardware. One of the major drawbacks of the Map-Reduce model is that,
to simplify reliability and fault tolerance, it does not preserve data
in memory between the map and reduce tasks of a Map-Reduce job or
across consecutive jobs, which imposes a high overhead to complex
workflows and graph algorithms, such as PageRank, which require
repetitive Map-Reduce jobs. Recent systems for cloud computing use
distributed memory for inter-node communication, such as the main
memory Map-Reduce (M3R~\cite{shinnar:vldb12}), Apache Spark~\cite{spark},
Apache Flink~\cite{flink}, Piccolo~\cite{piccolo:osdi10}, and distributed
GraphLab~\cite{graphlab:vldb12}. Another alternative framework to the
Map-Reduce model is the Bulk Synchronous Parallelism (BSP) programming
model~\cite{valiant:cacm90}. The best known implementations of the BSP
model for data analysis on the cloud are Google's
Pregel~\cite{pregel:podc09}, Apache Giraph~\cite{giraph}, and Apache
Hama~\cite{hama}.

Although the Map-Reduce framework was originally designed for batch
processing, there are several recent systems that have extended
Map-Reduce with online processing capabilities. Some of these systems
build on the well-established research on data streaming based on
sliding windows and incremental operators~\cite{babcock:pods02}, which
includes systems such as Aurora~\cite{aurora} and
Telegraph~\cite{telegraphCQ}. MapReduce Online~\cite{MRonline}
maintains state in memory for a chain of MapReduce jobs and reacts
efficiently to additional input records. It also provides a
memoization-aware scheduler to reduce communication across a
cluster. Incoop~\cite{incoop} is a Hadoop-based incremental processing
system with an incremental storage system that identifies the
similarities between the input data of consecutive job runs and splits
the input based on the similarity and file content. 
$i^2$MapReduce~\cite{i2mapreduce:tkde15} implements incremental iterative Map-Reduce jobs
using a store, MRB-Store, that maps input values to the reduce output values.
This store is used for detecting delta changes and propagating these changes to the output.
Google's Percolator~\cite{peng:osdi10} is a system based on BigTable for
incrementally processing updates to a large data set. It updates an
index incrementally as new documents are crawled. Microsoft
Naiad~\cite{naiad} is a distributed framework for cyclic dataflow
programs that facilitates iterative and incremental computations. It
is based on differential dataflow computations, which allow
incremental computations to have nested iterations.
CBP~\cite{logothetis:socc12} is a continuous bulk processing system on
Hadoop that provides a stateful group-wise operator that allows users
to easily store and retrieve state during the reduce stage as new data
inputs arrive. Their incremental computing PageRank implementation is
able to cut running time in half. REX~\cite{mihaylov:vldb12} handles
iterative computations in which changes in the form of deltas are
propagated across iterations and state is updated efficiently. In
contrast to our automated approach, REX requires the programmer to
explicitly specify how to process deltas, which are handled as first
class objects. Trill~\cite{trill} is a high throughput, low latency
streaming query processor for temporal relational data, developed at
Microsoft Research. The Reactive Aggregator~\cite{tangwongsan:vldb15},
developed at IBM Research, is a new sliding-window streaming engine
that performs many forms of sliding-window aggregation incrementally.
In addition to these general data analysis engines, there are many
data analysis algorithms that have been implemented incrementally,
such as incremental pagerank~\cite{desikan:www05}.\extended{
Finally, the
incremental query processing is related to the problem of incremental
view maintenance, which has been extensively studied in the context of
relational views (see~\cite{gupta:bde95} for a literature survey).

Many novel Big Data stream processing systems, also known as
distributed stream processing engines (DSPEs), have emerged recently.
The most popular one is Twitter's Storm~\cite{peng:osdi10}, which is
now part of the Apache ecosystem for Big Data analytics. It provides
primitives for transforming streams based on a user-defined topology,
consisting of spouts (stream sources) and bolts (which consume input
streams and may emit new streams). Other popular DSPE platforms
include Spark's D-Streams~\cite{spark:sosp13}, Flink
Streaming~\cite{flink}, Apache S4~\cite{S4}, and Apache
Samza~\cite{samza}.

In programming languages, self-adjusting
computation~\cite{acar:toplas09} refers to a technique for compiling
batch programs into programs that can automatically respond to changes
to their data. It requires the construction of a dependence graph at
run-time so that when the computation data changes, the output can be
updated by re-evaluating only the affected parts of the
computation. In contrast to our work, which requires only the state to
reside in memory, self-adjusting computation expects both the input
and the output of a computation to reside in memory, which makes it
inappropriate for unbounded data in a continuous stream.  Furthermore,
such dynamic methods impose a run-time storage and computation
overhead by maintaining the dependence graph. The main idea
in~\cite{acar:toplas09} is to manually annotate the parts of the input
type that is changeable, and the system will derive an incremental
program automatically based on these annotations. Each changeable
value is wrapped by a mutator that includes a list of reader closures
that need to be evaluated when the value changes. A read operation on
a mutator inserts a new closure, while the write operation triggers
the evaluation of the closures, which may cause writes to other
mutators, etc, resulting to a cascade of closure execution triggered
by changed data only. This technique has been extended to handle
incremental list insertions (like our work), but it requires the
rewriting of all list operations to work on incremental
lists. Recently, there is a proof-of-concept implementation of this
technique on map-reduce~\cite{acar:ddfp13}, but it was tested on a
serial machine. It is doubtful that such dynamic techniques can be
efficiently applied to a distributed environment, where a write in one
compute node may cause a read in another node. Finally, there is
recent work on static incrementalization based on
derivatives~\cite{cai:pldi14}. In contrast to our work, it assumes
that the merge function that combines the previous result with the
result on the delta changes uses exactly the same delta changes, a
restriction that excludes aggregations and group-bys.
}

\section{Earlier Work: MRQL}\label{mrql-language}

Apache MRQL~\cite{mrql} is a query processing and optimization system
for large-scale, distributed data analysis.  MRQL was originally
developed by the author~(\cite{edbt12,webdb11}), but is now
an Apache incubating project with many developers and users worldwide.
The MRQL language is an SQL-like query language for large-scale data
analysis on computer clusters. The MRQL query processing system can
evaluate MRQL queries in four modes: in Map-Reduce mode using Apache
Hadoop~\cite{hadoop}, in BSP mode (Bulk Synchronous Parallel model)
using Apache Hama~\cite{hama}, in Spark mode using Apache
Spark~\cite{spark}, and in Flink mode using Apache
Flink~\cite{flink}. The MRQL query language is powerful enough to
express most common data analysis tasks over many forms of raw in-situ
data, such as XML and JSON documents, binary files, and CSV documents.
The design of MRQL has been influenced by XQuery and ODMG OQL, although
it uses SQL-like syntax.  In fact, when restricted to XML, MRQL is as
powerful as XQuery.  MRQL is more powerful than other current
high-level Map-Reduce languages, such as Hive~\cite{hive} and
PigLatin~\cite{olston:sigmod08}, since it can operate on more complex
data and supports more powerful query constructs, thus eliminating the
need for using explicit procedural code. With MRQL, users are able to
express complex data analysis tasks, such as PageRank, k-means
clustering, matrix factorization, etc, using SQL-like queries
exclusively, while the MRQL query processing system is able to compile
these queries to efficient Java code that can run on various
distributed processing platforms.  For example, the PageRank query on
raw DBLP XML data, which ranks authors based on the number of citations they
have received from other authors, is 16 lines long~\cite{webdb11} and
can be executed on all the supported platforms as is, without changing
the query.

A recent extension to MRQL, called {\em MRQL Streaming}, supports the
processing of continuous MRQL queries over streams of batch data (that
is, data that come in continuous large batches).  Before the
incremental MRQL work presented in this paper, MRQL Streaming
supported traditional window-based streaming based on a fixed window
during a specified time interval.  Any batch MRQL query can be
converted to a window-based streaming query by replacing at least one
of the 'source' calls in the query that access data sources to
`stream' calls (with exactly the same call arguments).  For example,
the query:
\begin{lstlisting}
select (k,avg(p.Y))
from p in stream(binary,"points")
group by k: p.X
\end{lstlisting}
groups a stream of points by their $X$ coordinate and returns the
average $Y$ values in each group.  The MRQL Streaming engine
first processes all the existing sequence files in the directory
\s{points} and then checks this directory periodically for new
files. When new files are inserted in the directory, it processes
the new batch of data using distributed query processing.  MRQL
Streaming also supports a stream input format for listening to TCP
sockets for text input based on one of the MRQL Parsed Input Formats
(XML, JSON, CSV).  A query may work on multiple stream sources and
multiple batch sources. If there is at least one stream source, the
query becomes continuous (it never stops). The output of a continuous
query is stored in a file directory, where each file contains the
results of processing each batch of streaming data.  Currently, MRQL
Streaming works on Spark Streaming only but there are current efforts
to add support for Storm and Flink Streaming in the near future.  The
work reported here, called {\em Incremental MRQL}, extends the current
MRQL Streaming engine with incremental stream processing.

\section{The MRQL Algebra}\label{mrql-algebra}

Our compiler translates queries to algebraic terms and then uses
rewrite rules to put these algebraic terms into a homomorphic form,
which is then used to compute the query results incrementally by
combining the previous results with the results of processing the
incremental batches.

The MRQL algebra described in this section is a variation of the
algebra presented in our previous work~\cite{edbt12}, but is more
suitable for describing our incremental methods. The relational
algebra, the nested relational algebra, as well as many other database
algebras can be easily translated to our algebra. Our algebra
consists of a small number of higher-order homomorphic
operators~\cite{edbt12}, which are defined using structural recursion
based on the union representation of bags~\cite{tods00}. Monoid
homomorphisms capture the essence of many divide-and-conquer
algorithms and can be used as the basis for data
parallelism~\cite{tods00}.

The first operator, cMap (also known as concat-map or flatten-map in
functional programming languages), generalizes the select, project,
join, and unnest operators of the nested relational algebra. Given
two arbitrary types $\alpha$ and $\beta$, the operation
$\mathrm{cMap}(f,X)$ maps a bag $X$ of type $\{\alpha\}$ to a bag of
type $\{\beta\}$ by applying the function $f$ of type
$\alpha\rightarrow\{\beta\}$ to each element of $X$, yielding one bag
for each element, and then by merging these bags to form a single
bag of type $\{\beta\}$. Using a set former notation on bags, it is expressed as:
\begin{equation}\label{cmap-def}
\mathrm{cMap}(f,X)\;\defined\; \{\,z\,|\,x\in X,\,z\in f(x)\,\}
\end{equation}
or, alternatively, using structural recursion:
\[\begin{array}{rcl}
\mathrm{cMap}(f,X\uplus Y) & = & \mathrm{cMap}(f,X)\,\uplus\,\mathrm{cMap}(f,Y)\\
\mathrm{cMap}(f,\{a\}) & = & f(a)\\
\mathrm{cMap}(f,\{\,\}) & = & \{\,\}
\end{array}\]
Given an arbitrary type $\kappa$ that supports value equality ($=$),
an arbitrary type $\alpha$, and a bag $X$ of type $\{(\kappa,\alpha)\}$, the
operation $\mathrm{groupBy}(X)$ groups the elements of the bag $X$ by
their first component and returns a bag of type
$\{(\kappa,\{\alpha\})\}$.
For example, groupBy(\{ (1,``A"), (2,``B"), (1,``C") \})
returns \{(1,\{``A",``C"\}), (2,\{``B"\})\}.
The groupBy operation cannot be defined using a set former notation,
but can be defined using structural recursion:
\[\begin{array}{rcl}
\mathrm{groupBy}(X\uplus Y) & = & \mathrm{groupBy}(X)\,\gb_\uplus\,\mathrm{groupBy}(Y)\\
\mathrm{groupBy}(\{(k,a)\}) & = & \{(k,\{a\})\}\\
\mathrm{groupBy}(\{\,\}) & = & \{\,\}
\end{array}\]
where the parametric monoid $\gb_\oplus$ is a full outer join that merges groups associated with
the same key using the monoid $\oplus$ (equal to $\uplus$ for
groupBy):
\begin{eqnarray}
X\,\gb_\oplus\,Y & \defined & \{\,(k,a\oplus b)\,|\,(k,a)\in X,\,(k,b)\in Y\,\}\nonumber\\
&& \uplus\;\{\,(k,a)\,|\,(k,a)\in X,\,k\not\in\pi_1(Y)\,\}\label{gb-def}\\
&& \uplus\;\{\,(k,b)\,|\,(k,b)\in Y,\,k\not\in\pi_1(X)\,\}\nonumber
\end{eqnarray}
where $\pi_1(X)=\{\,k\,|\,(k,x)\in X\,\}$.
In other words, the monoid $\gb_\uplus$ constructs a set of pairs whose unique key is the first pair element.
In fact, any bag $X$ can be converted to a set
using $\pi_1(\mathrm{groupBy}(\mathrm{cMap}(\lambda x.\,\{(x,x)\},\,X)))$.
Note also that $\mathrm{groupBy}(X)$ is not the same as the nesting
$\{\,(k,\{\,y\,|\,(k',y)\in X,\,k=k'\,\})\,|\,(k,x)\in X\,\}$,
as the latter contains duplicate entries for the key $k$.
Unlike nesting, unnesting a groupBy returns the input bag (proven in \short{the extended version
of this paper~\cite{extended-stream}}\extended{Appendix~\ref{proofs}}):
\begin{equation}\label{groupby-unnest}
\{\,(k,v)\,|\,(k,s)\in\mathrm{groupBy}(X),\,v\in s\,\}\; =\; X
\end{equation}
\ignore{Other examples of $\gb_\oplus$ are key-value maps and vectors, constructed with the monoid $\gb_\rhd$,
where $x\rhd y=y$, that is, the new key value replaces the previous one.}

Although any join $X\bowtie_p Y$ can be expressed as a nested cMap:
\[\mathrm{cMap}(\lambda x.\,\mathrm{cMap}(\lambda y.\,\mathrm{if}\,p(x,y)\,\mathrm{then}\,\{(x,y)\}\,\mathrm{else}\,\{\,\},\,Y),\,X)\]
this term is not always a homomorphism on both inputs.
Instead, MRQL provides a special homomorphic operation for equi-joins and outer joins,
$\mathrm{coGroup}(X,Y)$, between a bag $X$ of type $\{(\kappa,\alpha)\}$
and a bag $Y$ of type $\{(\kappa,\beta)\}$ over their first component of a type $\kappa$,
which returns a bag of type $\{(\kappa,(\{\alpha\},\{\beta\}))\}$:
\[\begin{array}{rcl}
\mathrm{coGroup}(X_1\uplus X_2,Y_1\uplus Y_2) & = &\mathrm{coGroup}(X_1,Y_1)\\
&&\gb_{\uplus\times\uplus}\,\mathrm{coGroup}(X_2,Y_2)\\[1.5ex]
\mathrm{coGroup}(\{(k,a)\},\{(k,b)\}) & = & \{(k,(\{a\},\{b\}))\}\\[1.5ex]
\mathrm{coGroup}(\{(k,a)\},\{(k',b)\}) & = & \{(k,(\{a\},\{\,\})),\\
&&\skiptext{$\{$}(k',(\{\,\},\{b\}))\}\\[1.5ex]
\mathrm{coGroup}(\{(k,a)\},\{\,\}) & = & \{(k,(\{a\},\{\,\}))\}\\[1.5ex]
\mathrm{coGroup}(\{\,\},\{(k,b)\}) & = & \{(k,(\{\,\},\{b\}))\}\\[1.5ex]
\mathrm{coGroup}(\{\,\},\{\,\}) & = & \{(k,(\{\,\},\{\,\}))\}
\end{array}\]
where the product of two monoids, $\oplus\times\otimes$ is a monoid that,
when applied to two pairs $(x_1,x_2)$ and $(y_1,y_2)$, returns
$(x_1\oplus y_1,x_2\otimes y_2)$.
That is, the monoid $\gb_{\uplus\times\uplus}$ merges two bags of type $\{(\kappa,(\{\alpha\},\{\beta\}))\}$
by unioning together their $\{\alpha\}$ and $\{\beta\}$ values that correspond to the same key $\kappa$.
For example, coGroup(\{ (1,``A"), (2,``B"), (1,``C") \}, \{ (1,``D"), (2,``E"), (3,``F") \})
returns \{(1,(\{``A",``C"\},\{``D"\})), (2,(\{``B"\},\{``E"\})), (3,(\{\,\},\{``F"\}))\}.
It can be proven (with a proof similar to that of Equation~(\ref{groupby-unnest})) that both coGroup inputs can be derived from the coGroup result:
\[\begin{array}{lcl}
\{\,(k,x)\,|\,(k,(s_1,s_2))\in\mathrm{coGroup}(X,Y),\,x\in s_1\,\}  & = & X\\
\{\,(k,y)\,|\,(k,(s_1,s_2))\in\mathrm{coGroup}(X,Y),\,y\in s_2\,\}  & = & Y
\end{array}\]
and a coGroup is equivalent to a groupBy if one of the inputs is empty:
\[\begin{array}{lcl}
\mathrm{coGroup}(X,\{\,\})& = & \{\,(k,(s,\{\,\}))\,|\,(k,s)\in\mathrm{groupBy}(X)\,\}
\end{array}\]

Aggregations are captured by the operation $\mathrm{reduce}(\oplus,X)$,
which aggregates a bag $X$ using a commutative monoid $\oplus$. For
example, $\mathrm{reduce}(+,\{1,2,3\})=6$. This operation
is in fact a general homomorphism that can be defined on any monoid $\oplus$
with an identity function. We have:
\[\begin{array}{rcl}
\mathrm{reduce}(\uplus,X) & = & X\\
\mathrm{reduce}(\gb_\uplus,X) & = & \mathrm{groupBy}(X)
\end{array}\]
Finally, iteration
$\mathrm{repeat}(F,n,X)$ over the bag $X$ of type $\{\alpha\}$ applies
$F$ of type $\{\alpha\}\rightarrow\{\alpha\}$ to $X$ $n$ times, yielding a bag
of type $\{\alpha\}$:
\[\mathrm{repeat}(F,n,X)=F^n(X)=F(F(\ldots F(X)))\]
An iteration is a homomorphism as long as $F$ is a homomorphism,
that is, when $F(X\uplus Y)=F(X)\uplus F(Y)$.

\begin{definition}[MRQL Algebra]\label{algebra}
The MRQL algebra consists of terms that take the following form:
\[\begin{array}{llll}
e,e_1,e_2 & ::= & \mathrm{cMap}(f,e) & \mbox{flatten-map}\\
&|&\mathrm{groupBy}(e) & \mbox{group-by}\\
&|&\mathrm{coGroup}(e_1,e_2) & \mbox{join}\\
&|&\mathrm{reduce}(\oplus,e) & \mbox{aggregation}\\
&|& S_i & \mbox{stream source}
\end{array}\]
where $\oplus$ is a monoid on basic types, such as $+$, $*$, $or$, $and$, etc.
Function $f$ is an anonymous function that may contain such algebraic terms but is not
permitted to contain any reference to a stream source, $S_i$.
\end{definition}

\noindent
The restriction on $f$ in Definition~\ref{algebra} excludes non-equi-joins, such as cross products, which require
a nested cMap in which the inner cMap is over a data source.
This algebra does not include iteration, $\mathrm{repeat}(F,n,e)$.
Iterations have been discussed in Section~\ref{approach}.
In addition, for brevity, this algebra does not include terms for non-streaming
input sources, general tuple and record construction and projection, bag union, singleton
and empty bag, arithmetic operations, if-then-else expressions,
boolean operations, etc.

In addition to these operations, there are a few more algebraic
operations that can be expressed as homomorphisms, such as
$\mathrm{orderBy}(X)$, which is a groupBy followed by a sorting over
the group-by key. This operation returns a list, which is represented
by the non-commutative monoid list-append, \app. Mixing multiple
collection monoids and operations in the same algebra has been
addressed by our previous work~\cite{tods00}, and is left out from
this paper to keep our analysis simple.

For example, the query $q(S_1,S_2)$ used as the first example in Section~\ref{approach},
is translated to the following algebraic term $Q(S_1,S_2)$:
\[\begin{array}{l}
\mathrm{cMap}(\lambda(k,s).\,\{(k,avg(s))\},\\
\skiptext{$\mathrm{cM}($}\mathrm{groupBy}(\mathrm{cMap}(\lambda(j,(xs,ys)).\,g(xs,ys),\,\\
\skiptext{$\mathrm{cM}(\mathrm{groupBy}(\mathrm{cMap}($}\mathrm{coGroup}(\mathrm{cMap}(\lambda x.\,\{(x.B,x)\},\,S_1),\\
\skiptext{$\mathrm{cM}(\mathrm{groupBy}(\mathrm{cMap}(\mathrm{coGroup}($}\mathrm{cMap}(\lambda y.\,\{(y.C,y)\},\,S_2)))))
\end{array}\]
where $avg(s)=\mathrm{reduce}(+,s)/\mathrm{reduce}(+,\mathrm{cMap}(\lambda v.\,\{1\},s))$\linebreak
and $g(xs,ys)=\mathrm{cMap}(\lambda x.\,\mathrm{cMap}(\lambda y.\,\{(x.A,y.D)\},\,ys),\,xs)$.

Although all algebraic operators used in MRQL are homomorphisms, their
composition may not be. For instance,
$\mathrm{cMap}(f,\mathrm{groupBy}(X))$ is not a homomorphism for
certain functions $f$, because, in general, cMap does not distribute
over $\gb_\uplus$. One of our goals is to transform any composition of
algebraic operations into a homomorphism.

\section{Query Normalization}\label{normalization}

In an earlier work~\cite{edbt12}, we have presented general algorithms for unnesting
nested queries. For example, consider the following nested query:
\begin{lstlisting}
select x from x in X
where x.D > sum(select y.C from y in Y where x.A=y.B)
\end{lstlisting}
A typical method for evaluating this query in a relational system is to first
group \s{Y} by \s{y.B}, yielding pairs of \s{y.B} and \s{sum(y.C)},
and then to join the result with \s{X} on \s{x.A=y.B} using a
left-outer join, removing all those matches whose \s{x.D} is below
the sum. But, in our framework, this query is translated into:

\begin{tab}
cMap( \=\+$\lambda$(k,(xs,ys)).\,cMap( \=\+$\lambda$x.\,\=\+{\bf if} x.D $>$ reduce(+,ys)\\
 {\bf then} \{x\}\\
 {\bf else} \{\,\},\, xs),\-\-\\
coGroup( \=\+cMap( $\lambda$x.\,\{(x.A,x)\}, X ),\\
         cMap( $\lambda$y.\,\{(y.B,y.C)\}, Y ) ) )
\end{tab}

\noindent
That is, the query unnesting is done with a left-outer join, which is
captured concisely by the coGroup operation without the need for using
an additional group-by operation or handling null values. This
unnesting technique was generalized to handle arbitrary nested
queries, at any place, number, and nesting level (the reader is
referred to our earlier work~\cite{edbt12} for details).

The algebraic terms derived from MRQL queries can be normalized using the following rule:
\begin{equation}
\mathrm{cMap}(f,\mathrm{cMap}(g,S)) \rightarrow \mathrm{cMap}(\lambda x.\,\mathrm{cMap}(f,g(x)),S)\label{cmap-opt}
\end{equation}
which fuses two cascaded cMaps into a nested cMap, thus avoiding 
the construction of the intermediate bag.
This rule can be proven directly from the cMap definition in Equation~(\ref{cmap-def}):
\[\begin{array}{l}
\mathrm{cMap}(f,\mathrm{cMap}(g,S))\\
=\;\{\,z\,|\, w\in\{\,y\,|\,x\in S,\,y\in g(x)\,\},\, z\in f(w)\,\}\\
=\;\{\,z\,|\, x\in S,\,y\in g(x),\, z\in f(y)\,\}\\
=\;\{\,z\,|\, x\in S,\,z\in\{\,w\,|\,y\in g(x),\, w\in f(y)\,\}\,\}\\
=\;\mathrm{cMap}(\lambda x.\,\mathrm{cMap}(f,g(x)),S)
\end{array}\]
If we apply the transformation~(\ref{cmap-opt}) repeatedly, 
and given that we can always use the identity
$\mathrm{cMap}(\lambda x.\{x\},X)=X$ in places where there is no cMap between groupBy/coGroup operations,
any algebraic terms in Definition~\ref{algebra}
can be normalized into the following form:

\begin{definition}[Normalized MRQL Algebra]\label{normalized-algebra}
The normalized MRQL algebra consists of terms $q$ that take the following form:
\[\begin{array}{rlll}
q,q_1,q_2 & ::= & 
\mathrm{reduce}(\oplus,\,c) &\mbox{the query header}\\
&|&c\\[1ex]
c,c_1,c_2 & ::= & \mathrm{cMap}(f,\,e)\\[1ex]
e & ::= & \mathrm{groupBy}(c) & \mbox{the query body}\\
&|&\mathrm{coGroup}(c_1,\,c_2)\\
&|& S_i
\end{array}\]
where function $f$ is an anonymous function that does not contain any reference to a stream source,
$S_i$. 
\end{definition}

\noindent
The query body is a tree of groupBy/coGroup operations connected via
cMaps.

\begin{figure*}
\setlength{\jot}{2ex}
\begin{subequations}
\leoframe{\begin{minipage}[l]{0.4\linewidth}
\vspace*{-2ex}\begin{gather}
\rho\vdash v:\rho(v)\label{bind}\\
\frac{\rho\vdash e:\Box}{\forall v\in e:\,\rho\vdash v:\Box}\label{box}\\
\frac{\rho\vdash\mathrm{reduce}(\oplus,X):\oplus}{\rho\vdash X:\uplus}\label{reduce}\\
\frac{\rho\vdash (X\uplus Y):\uplus}{\rho\vdash X:\uplus,\;\rho\vdash Y:\uplus}
\end{gather}
\end{minipage}
\begin{minipage}[l]{0.55\linewidth}
\begin{gather}
\frac{\rho\vdash\mathrm{cMap}(f,X):\uplus}{\rho\vdash X:\uplus}\label{cmap-box}\\
\frac{\rho\vdash\mathrm{groupBy}(X):\,\gb_\uplus}{\rho\vdash X:\uplus}\label{groupby-monoid}\\
\frac{\rho\vdash\mathrm{coGroup}(X,Y):\,\gb_{\uplus\times\uplus}}{\rho\vdash X:\uplus,\;\rho\vdash Y:\uplus}\label{cogroup-monoid}\\
\frac{\rho\vdash\mathrm{cMap}(f,X):\,\gb_\otimes}{\rho\vdash X:\,\gb_\oplus,\;\rho[k:\Box,s:\oplus]\vdash f(k,s):\,\gb_\otimes}\label{cmap-groupby}\\[-2ex]\nonumber
\end{gather}
\end{minipage}}
\end{subequations}
\setlength{\jot}{1ex}
\caption{Monoid Inference Rules}\label{monoid-infer}
\end{figure*}

\section{Monoid Inference}\label{monoid-inference}

One of our tasks is, given an algebraic term
$f(\overline{S})$, where an $S_i\in\overline{S}$ is a
streaming data source, to prove that $f$ is a homomorphism by deriving a monoid
$\otimes$ such that:
\begin{equation}
f(S_1\uplus S'_1,\ldots,S_n\uplus S'_n)=f(S_1,\ldots,S_n)\otimes f(S'_1,\ldots,S'_n)\label{monoid-eq}
\end{equation}
We have developed a monoid inference system, inspired by type inference systems
used in programming languages.
We use the judgment $\rho\vdash e:\oplus$ to indicate that $e$ is
a monoid homomorphism with a merge function $\oplus$ under the
environment $\rho$, which binds variables to monoids.
The notation
$\rho(v)$ extracts the binding of the variable $v$, while
$\rho[v:\oplus]$ extends the environment with a new binding from $v$
to $\oplus$. 
Equation~(\ref{monoid-eq}) can now be expressed as the judgment:
\[[S_1\!:\uplus,\ldots,S_n\!:\uplus]\vdash f(\overline{S}):\otimes\]
If a term is invariant under change, such as an invariant data source,
it is associated with the special monoid $\Box$:
\[X\,\Box\,Y\; \defined\;
\left\{\begin{array}{ll}
X & \mbox{if $X=Y$}\\
\mbox{error} & \mbox{otherwise}
\end{array}\right.\]

Our monoid inference algorithm is a heuristic algorithm that annotates
terms with monoids (when possible). It is very similar to type inference. Most of our inference
rules are expressed as fractions: the denominator (below the line) contains the premises
(separated by comma) and the numerator (above the line) is the conclusion. For
example, the rule $\frac{\rho\,\vdash f(x):\,\otimes}{\rho\,\vdash
x:\,\oplus}$ indicates that $f(x_1\oplus x_2)=f(x_1)\otimes f(x_2)$.
Figure~\ref{monoid-infer} gives some of the inference rules. More
rules will be given in Lemma~\ref{smap-hom}.
Rules~(\ref{reduce}) through~(\ref{groupby-monoid}) are derived directly from the
algebraic definition of the operators. 
Rule~(\ref{bind}) retrieves the associated monoid of a variable $v$ from the environment $\rho$.
Rule~(\ref{box}) indicates that if all the variables in a term $e$ are invariant, then so is $e$.
Rule~(\ref{cmap-groupby}) indicates that a cMap over a groupBy is a homomorphism
as long as its functional argument is a homomorphism. It can be proven as follows:
\[\begin{array}{l}
\mathrm{cMap}(f,\,X\gb_\oplus Y)\\
=\{\,(\theta,z)\,|\,(k,s)\in(X\gb_\oplus Y),\,(\theta,z)\in f(k,s)\,\}\\
=\{\,(\theta,z)\,|\,(k,x)\in X,\,(k,y)\in Y,\,(\theta,z)\in f(k,x\oplus y)\,\}\uplus\cdots\\
=\{\,(\theta,z)\,|\,(k,x)\in X,\,(k,y)\in Y,\\
\skiptext{$=\{\,(\theta,z)\,|\,$}(\theta,z)\in (f(k,x)\gb_\otimes f(k,y))\,\}\uplus\cdots\\
=\{\,(\theta,z)\,|\,(k,x)\in X,\,(k,y)\in Y,\\
\skiptext{$=\{\,(\theta,z)\,|\,$}(\theta,z)\in\{\,(\theta,a\otimes b)\,|\,(\theta,a)\in f(k,x),\\
\skiptext{$=\{\,(\theta,z)\,|\,(\theta,z)\in\{\,(\theta,a\otimes b)\,|\,$}(\theta,b)\in f(k,y)\,\}\,\}\uplus\cdots\\
=\{\,(\theta,a\otimes b)\,|\,(k,x)\in X,\,(k,y)\in Y,\,(\theta,a)\in f(k,x),\\
\skiptext{$=\{\,(\theta,a\otimes b)\,|\,$}(\theta,b)\in f(k,y)\,\}\uplus\cdots\\
=\{\,(\theta,a\otimes b)\,|\,(\theta,a)\in\mathrm{cMap}(f,X),\,(\theta,b)\in\mathrm{cMap}(f,Y)\,\}\\
\skiptext{$=\{\,$}\uplus\cdots\;\hspace*{5ex}\mbox{(given that $k$ and $\theta$ are unique keys)}\\
=\mathrm{cMap}(f,X)\,\gb_\otimes\,\mathrm{cMap}(f,Y)
\end{array}\]

\section{Injecting Lineage Tracking}\label{injection}

There are two tasks that need to be accomplished to achieve our goal of
transforming an algebraic term into a homomorphism: 1) transform the
algebraic term in such a way that it propagates all keys used in joins and group-bys to
the query output, and 2) pull the non-homomorphic parts out of the algebraic term
so that it becomes a homomorphism. In this section, we address
the first task.

We will transform the algebraic terms given in
Definition~\ref{normalized-algebra} in such a way that they propagate
the join and the group-by keys. That is, each value $v$ returned by
these terms is annotated with a lineage $\theta$, as a pair
$(\theta,v)$, where $\theta$ takes the following form:
\ignore{
\[\begin{array}{rlll}
\theta,\theta_1,\theta_2 & ::= & (\theta_1,\theta_2)\\
       & | & k & \mbox{a groupBy or coGroup key}\\
       & | & () & \mbox{empty lineage}
\end{array}\]
}
\[\begin{array}{rlll}
\theta,\theta_1,\theta_2 & ::= & (\kappa,\theta) & \mbox{extended with groupBy key $\kappa$}\\
       & | & (\kappa,(\theta_1,\theta_2)) & \mbox{extended with coGroup key $\kappa$}\\
       & | & () & \mbox{empty lineage}
\end{array}\]
That is, the lineage $\theta$ of the query result $v$ is the tree of
the groupBy and coGroup keys that are used in deriving the result $v$
(one key for each groupBy and coGroup operation).  The lineage tree
$\theta$ has the same shape as the groupBy/coGroup tree of the query.
We transform a query $q$ in such a way that, if the output of the query is
$\{t\}$ for some type $t$, then the transformed query will have output
$\{(\theta,t)\}$.  Furthermore, if the the output is a
non-collection type $t$, then the transformed query will also have
output $\{(\theta,t)\}$, which separates the contributions to $t$
associated with each combination of group-by/join keys.

\begin{algorithm}[Lineage Annotation]\label{annotation-alg}\ \\
{\bf Input:} a normalized query term $q$ defined in Definition~\ref{normalized-algebra}\\
{\bf Output:} a term $\trq{q}$ annotated with a lineage $\theta$

\noindent
\framebox{\hspace*{-1.5ex}\parbox[t]{\columnwidth}{
\begin{subequations}
\begin{align}
\setlength{\jot}{1.5ex}
\trq{\mathrm{reduce}(\oplus,\,\mathrm{cMap}(f&,\,S_i))}\nonumber\\[-1ex]
 = \{((),&\,\mathrm{reduce}(\oplus,\,\mathrm{cMap}(f,\,S_i)))\}\label{norm2a}\\
\trq{\mathrm{reduce}(\oplus,\,\mathrm{cMap}(f&,\,e))}\nonumber\\[-1ex]
 = \mathrm{reduce}&(\gb_\oplus,\,\mathrm{sMap1}(f,\,\tre{e}))\label{norm2}\\
\trq{\mathrm{cMap}(f,\,S_i)} &= \{\,((),b)\,|\,a\in S_i,\,b\in f(a)\,\}\label{norm3}\\
\trq{\mathrm{cMap}(f,\,e)} & = \mathrm{sMap1}(f,\,\tre{e})\label{norm4}\\
\tre{\mathrm{groupBy}(c)} & = \mathrm{groupBy}(\mathrm{swap}(\trc{c}))\label{norm5}\\
\tre{\mathrm{coGroup}(c_1,c_2)} = & \;\mathrm{mix}(\mathrm{coGroup}(\trc{c_1},\trc{c_2}))\label{norm6}\\
\trc{\mathrm{cMap}(f,\,S_i)} & = \mathrm{sMap3}(f,\,S_i)\label{norm7}\\
\setlength{\jot}{1ex}
\trc{\mathrm{cMap}(f,\,e)} & = \mathrm{sMap2}(f,\,\tre{e})\label{norm8}
\end{align}
\end{subequations}}}\ \\[1ex]
where sMap1, sMap2, sMap3, swap, and mix are defined as follows:
\begin{subequations}
\begin{align}
\mathrm{sMap1}(f,X) &\defined \{\,((\theta,k),b)\,|\,((\theta,k),a)\in X,\nonumber\\
&\skiptext{$\defined \{\,((\theta,k),b)\,|\;\;$}b\in f(k,a)\,\}\label{smap1}\\
\mathrm{sMap2}(f,X) &\defined \{\,(k',((k,\theta),b))\,|\,((k,\theta),a)\in X,\nonumber\\
&\skiptext{$\defined \{\,(k',((k,\theta),b))\,|\;\;$}(k',b)\in f(k,a)\,\}\label{smap2}\\
\mathrm{sMap3}(f,X) &\defined \{\,(k,((),b))\,|\,a\in X,\,(k,b)\in f(a)\,\}\label{smap3}\\
\mathrm{swap}(X) &\defined \{\,((k,\theta),v)\,|\,(k,(\theta,v))\in X\,\}\label{swap}\\
\mathrm{mix}(X) &\defined \{\,((k,(\theta_x,\theta_y)),(xs,ys))\nonumber\\
&\skiptext{$\defined\;\;$}|\,(k,(s_1,s_2))\in X,\label{mix}\\
&\skiptext{$\defined\;|$}(\theta_x,xs)\in\mathrm{groupBy}(s_1),\nonumber\\
&\skiptext{$\defined\;|$}(\theta_y,ys)\in\mathrm{groupBy}(s_2)\,\}\nonumber
\end{align}
\end{subequations}
\end{algorithm}

\noindent
A query $q$ in our framework is transformed in such a way that it
propagates the lineage from the data stream sources to the query
output, starting with the empty lineage $()$ at the sources and
extended with the join and group-by keys.
The sMap1 operation is a cMap that propagates the input lineage
$\theta$ to the output as is.  The sMap2 operation is a cMap that
extends the input lineage $\theta$ with a groupBy/coGroup key $k$.
The lineage propagation is done by the cMap Rules~(\ref{norm3})
and~(\ref{norm8}).  Rule~(\ref{norm3}) applies to the outer query cMap
that produces the query output. It simply propagates the lineage from
the cMap input to the output. Rule~(\ref{norm8}) applies to a cMap
that returns the input of a groupBy or coGroup. The output of this
cMap must be a bag of key-value pairs, as is expected by a groupBy or
a coGroup. Thus, Rule~(\ref{norm8}) extends the lineage with a new key
and prepares the cMap output for the enclosing groupBy or
coGroup. This is done by translating cMap to sMap2.
Rule~(\ref{norm2}) indicates that a total aggregation becomes a
group-by aggregation by aggregating the values of each group
associated with a different lineage $\theta$.  Rule~(\ref{norm5})
translates a groupBy on a key to a groupBy on the entire lineage
(which includes the groupBy key).  Rule~(\ref{norm6}) translates a
coGroup on a key to a coGroup on the entire lineage, but it is done
using the function mix (defined in~(\ref{mix}) because the left input
lineage $\theta_x$ is not necessarily compatible with the right input
lineage $\theta_y$.  Given this, Rule~(\ref{norm6}) generates a
coGroup on the join key first, and then, for each join key $k$, it
groups the left and right join matches by $\theta_x$ and $\theta_y$
respectively, so that the output contains unique lineage key
combinations, $(k,(\theta_x,\theta_y))$.  Finally, Rule~(\ref{norm7})
annotates each value of the input stream $S_i$ with the empty lineage
$()$.

We will prove next that the transformed query $\trq{q}$ is a homomorphism.
We first prove that the generated sMap1 and sMap2 in $\trq{q}$ are homomorphisms, provided that
their functional arguments are homomorphisms. More specifically, we
prove the following judgments:
\begin{lemma}[Transformed Term Judgments]\label{smap-hom}
\setlength{\jot}{2ex}
\begin{subequations}
\begin{gather}
\frac{\rho\,\vdash\,\mathrm{sMap1}(f,\,X):\,\gb_\otimes}
{\rho\,\vdash X:\,\gb_\oplus,\,\rho[k:\Box,v:\oplus]\vdash f(k,v):\,\gb_\otimes}\label{smap1-rule}\\
\frac{\rho\vdash\mathrm{sMap2}(f,X):\uplus}
{\rho\,\vdash X:\,\gb_\oplus,\,\rho[k:\Box,v:\oplus]\vdash f(k,v):\uplus}\label{smap2-rule}\\
\frac{\rho\vdash\mathrm{sMap3}(f,X):\uplus}
{\rho\vdash X:\uplus}\label{smap3-rule}\\
\frac{\rho\vdash\mathrm{swap}(X):\uplus}
{\rho\,\vdash X:\uplus}\label{swap-rule}\\
\frac{\rho\,\vdash\,\mathrm{mix}(X):\,\gb_{\uplus\times\uplus}}
{\rho\,\vdash X:\,\gb_{\uplus\times\uplus}}\label{mix-rule}
\end{gather}
\end{subequations}
\setlength{\jot}{1ex}
\end{lemma}
\noindent
The proof of this Lemma is given in \short{the extended version
of this paper~\cite{extended-stream}}\extended{Appendix~\ref{proofs}}.

\begin{definition}[Query Merger Monoid]\label{merger-monoid}
The query merger monoid $\monoid{q}$ of a normalized query term $q$
is defined as follows:
\[\begin{array}{rcl}
\monoid{\mathrm{reduce}(\oplus,\,\mathrm{cMap}(f,\,e))} & = & \gb_\oplus\\
\monoid{\mathrm{cMap}(f,\,S_i)} & = & \uplus\\
\monoid{\mathrm{cMap}(f,\,e)} & = & \gb_\otimes\hspace*{5ex}\mbox{if $e\not= S_i$}
\end{array}\]
where the monoid $\gb_\otimes$ in the last equation comes from\linebreak
$\mathrm{sMap1}(f,\,X):\;\gb_\otimes$ in Equation~(\ref{smap1-rule}).
\end{definition}
\noindent
Based on Lemma~\ref{smap-hom}, we can now prove that the transformed query is a homomorphism:
\begin{theorem}[Homomorphism]\label{hom}
Any transformed term $\trq{q}$, where $q$ is defined in
Definition~\ref{normalized-algebra}, is a homomorphism over the input
streams, provided that each generated sMap1 and sMap2 term in $\trq{q}$ satisfies
the premises in Judgment~(\ref{smap1-rule}) and~(\ref{smap2-rule}):
\begin{gather}
\rho[S_1:\uplus,\ldots,S_n:\uplus]\vdash\trq{q}:\;\monoid{q}\label{cc2}
\end{gather}
\end{theorem}
\noindent
The proof is given in \short{the extended version
of this paper~\cite{extended-stream}}\extended{Appendix~\ref{proofs}}.
\begin{definition}[Query Answer]\label{query-answer}
The query answer $\answer{q}_x$ of the query $q$,
defined in Definition~\ref{normalized-algebra},
is a function over the current state $x=\trq{q}$ and is derived as follows:
\[\begin{array}{l}
\answer{\mathrm{cMap}(f,\,S_i)}_x = \pi_2(x)\\
\answer{\mathrm{reduce}(\oplus,\,\mathrm{cMap}(f,\,e))}_x = \mathrm{reduce}(\oplus,\,\pi_2(x))\\
\answer{\mathrm{cMap}(f,\,e)}_x = \pi_2(\mathrm{reduce}(\gb_\otimes,T(x)))\hspace*{5ex}\mbox{if $e\not= S_i$}\\
\answer{q}_x = \pi_2(\mathrm{elem}(\mathrm{reduce}(\gb_\oplus,x)))\hfill\mbox{otherwise}
\end{array}\]
where the monoid $\oplus$ in the last equation is $\monoid{q}$,
the monoid $\otimes$ in the second equation comes from
$\mathrm{sMap1}(f,\,X):\;\gb_\otimes$ in Equation~(\ref{smap1-rule}), and
\[\begin{array}{rcl}
\mathrm{elem}(X) & = & \left\{\begin{array}{ll}
v & \mbox{if $X=\{v\}$}\\
\mathrm{error} & \mbox{otherwise}
\end{array}\right.\\[2.5ex]
T(X) & = & \{\,(k,a)\,|\,((k,\theta),a)\in X\,\}
\end{array}\]
\end{definition}
\noindent
The following theorem proves that $\answer{q}_x$ returns the same answer as $q$:
\begin{theorem}[Correctness]\label{correctness}
The $\answer{q}_x$ over the state $x=\trq{q}$ returns the same result as the original query $q$,
where $q$ is defined in Definition~\ref{normalized-algebra}:
\begin{align}
&\answer{q}_x = q & \mbox{for $x=\trq{q}$}\label{cc1}
\end{align}
\end{theorem}
\noindent
The proof of this theorem is given in \short{the extended version
of this paper~\cite{extended-stream}}\extended{Appendix~\ref{proofs}}.

\subsection{Restricting Joins}

Theorem~\ref{hom} indicates that a query transformed by
Algorithm~\ref{annotation-alg} is a homomorphism if the cMap
functional arguments in the query satisfy the premises in
Judgment~(\ref{smap1-rule}) and~(\ref{smap2-rule}).  But, consider the
following join:
\[\begin{array}{l}
\mathrm{cMap}(\lambda(k,(xs,ys)).\,\mathrm{cMap}(\lambda x.\,\mathrm{cMap}(\lambda y.\,\{(x,y)\},\,xs),\,ys),\\
\skiptext{$\mathrm{cMap}($}\mathrm{coGroup}(X,Y))
\end{array}\]
Our monoid inference system can not prove that the functional
argument of this cMap is a homomorphism on both $xs$ and $ys$.  This is
expected because this join could be a many-to-many join, which we know
can not be a homomorphism.  Since we have given up on handling many-to-many
joins, we want to extend the monoid inference algorithm to always
assume that every join is a one-to-one or one-to-many join and draw
inferences based on this assumption.

We have already seen in Section~\ref{monoid-inference} that
if a term is invariant under change, such as an invariant data source,
it is associated with the special monoid $\Box$.
We can also use the annotation $\Box$ to denote certain functional
dependencies, such as $\kappa\rightarrow\alpha$ on a bag $X$ of type
$\{(\kappa,\alpha)\}$, which indicates that the second component of a
pair in $X$ depends on the first. This dependency is captured by
annotating $\mathrm{groupBy}(X)$ with the monoid $\gb_\Box$, which
indicates that each group remains invariant under change, implying
that the group-by key is also a unique key of $X$. That is, if
$(k,s_1)\in X_1$ and $(k,s_2)\in X_2$, then $\{\,(k,s_1\Box
s_2)\,|\,(k,s_1)\in X_1,\,(k,s_2)\in X_2\,\}$ in
Definition~(\ref{gb-def}) must have $s_1=s_2$ so that $s_1\Box
s_2=s_1$, otherwise it will be an error. This means that we can not
have two different groups associated with the same key $k$. Given that
a groupBy over a singleton gives a singleton group, each group
returned from a groupBy is a singleton that remains invariant.
A similar functional dependency can also apply to a join between $X$
of type $\{(\kappa,\alpha)\}$ and $Y$ of type $\{(\kappa,\beta)\}$,
which is a $\mathrm{coGroup}(X,Y)$ of type
$\{(\kappa,(\{\alpha\},\{\beta\}))\}$. To indicate that this join is
one-to-one or one-to-many, we annotate $\mathrm{coGroup}(X,Y)$ with
the monoid $\gb_{\Box\times\uplus}$, which enforces the constraint that the bag
$\{\alpha\}$ in $\{(\kappa,(\{\alpha\},\{\beta\}))\}$ be either empty
or singleton, that is, at most one $X$ element can be joined with $Y$
over a key $\kappa$. This is because the $X$ and $Y$ values that are
joined over the same key are merged with $\Box$ and $\uplus$,
respectively, as indicated by $\gb_{\Box\times\uplus}$.
Therefore, to incorporate the assumption 
that all joins are one-to-one or one-to-many in the monoid inference system,
we have to replace Judgment~(\ref{cogroup-monoid}) with the following judgment:
\begin{gather}
\frac{\rho\vdash\mathrm{coGroup}(X,Y):\,\gb_{\Box\times\uplus}}{\rho\vdash X:\uplus,\;\rho\vdash Y:\uplus}\label{cogroup-box}
\end{gather}
Given this judgment, the cMap input of our join example will be annotated
with $\gb_{\Box\times\uplus}$, which means that $xs$ and $ys$ will be annotated with $\Box$ and $\uplus$, respectively
(ie, $xs$ is invariant). Based on these annotations, the functional argument
$\mathrm{cMap}(\lambda x.\,\mathrm{cMap}(\lambda y.\,\{(x,y)\},\,xs),\,ys)$ satisfies
the premise in Judgment~(\ref{smap2-rule}) since it is annotated with $\uplus$.

\subsection{Handling Iterative Queries}\label{iteration}

The approximate algorithm that processes iterations incrementally,
presented in Section~\ref{approach}, requires an additional operation,
called the diffusion operator $\diffuse{\otimes}$, that satisfies the property:
\[T\otimes(T\diffuse{\otimes}\Delta T)=T\otimes(T\otimes\Delta T)\]
Our goal is to define $\diffuse{\otimes}$ in terms of $\otimes$ in such a way that
$T\diffuse{\otimes}\Delta T$ contains only those parts of $T\otimes\Delta T$
that changed from $T$.
\begin{definition}[Diffusion] The diffusion $\diffuse{\otimes}$ of
a monoid $\otimes$ is defined as follows:
\[\begin{array}{rcllrcl}
\diffuse{\gb_\oplus} & = & \Downarrow_{\diffuse{\oplus}}&&
\diffuse{\otimes\times\oplus} & = & \diffuse{\otimes}\times\diffuse{\oplus}\\
\diffuse{\oplus} & = & \multicolumn{4}{l}{\oplus\hspace*{4ex}\mbox{for any other monoid}}
\end{array}\]
where $\Downarrow_{\otimes}$ is a right-outer join defined as follows:
\[\begin{array}{rcl}
X\,\Downarrow_\otimes\,Y & \defined & \{\,(k,a\otimes b)\,|\,(k,a)\in X,\,(k,b)\in Y\,\}\\
&&\uplus\;\{\,(k,b)\,|\,(k,b)\in Y,\,k\not\in\pi_1(X)\,\}
\end{array}\]
\end{definition}

\noindent
This $\diffuse{\otimes}$ satisfies the
desired property, $T\otimes(T\diffuse{\otimes}\Delta T)=T\otimes(T\otimes\Delta T)$
(proven in \short{the extended version
of this paper~\cite{extended-stream}}\extended{Appendix~\ref{proofs}}).

\section{Non-Homomorphic Terms}\label{factoring}

Theorem~\ref{correctness} indicates that the terms generated by the
transformations~(\ref{norm2a}) through~(\ref{norm8}) are homomorphisms
as long as the premises of the judgements in Lemma~\ref{smap-hom}
are true. These premises indicate that the cMap functional arguments
must be homomorphisms too. We want the operations that cannot be
annotated with a monoid to be transformed so that the non-homomorphic
parts of the operation are pulled outwards from the query using
rewrite rules. We achieve this with the help of kMap:
\[\mathrm{kMap}(f,X)\;\defined\;\mathrm{cMap}(\lambda (\theta,v).\,\{(\theta,f(v))\},\,X)\]
which is a cMap that propagates the lineage $\theta$ as is.
In our framework, all non-homomorphic parts take the form of a kMap
and are accumulated into one kMap using
\[\mathrm{kMap}(f,\mathrm{kMap}(g,X))\;\rightarrow\;\mathrm{kMap}(\lambda x.\,f(g(x)),\,X)\]
More specifically, we first split each non-homomorphic cMap to a
composition of a kMap and a cMap so that the latter cMap is a
homomorphism, and then we pull and merge kMaps. Consider the term
$\mathrm{cMap}(\lambda (\theta,v).\,\{(\theta',e)\},\,X)$, which
creates a new lineage $\theta'$ from the old $\theta$. In our
framework, we find the largest subterms in the algebraic term $e$,
namely $e_1,\ldots,e_n$, that are homomorphisms. This is accomplished
by traversing the tree that represents the term $e$, starting from the
root, and by checking if the node can be inferred to be a homomorphism. If
it is, the node is replaced with a new variable. Thus, $e$ is mapped
to a term $f(e_1,\ldots,e_n)$, for some term $f$, and the terms
$e_1,\ldots,e_n$ are replaced with variables when $f$ is pulled
outwards:
\[\begin{array}{l}
\mathrm{cMap}(\lambda (\theta,v).\,\{(\theta',e)\},\,X)\\
\hspace*{8ex}\rightarrow \mathrm{cMap}(\lambda (\theta,v).\,\{(\theta',f(e_1,\ldots,e_n))\},\,X)\\
\hspace*{8ex}\rightarrow \mathrm{kMap}(\lambda (x_1,\ldots,x_n).\,f(x_1,\ldots,x_n),\\
\hspace*{8ex}\skiptext{$\rightarrow \mathrm{kMap}($}\mathrm{cMap}(\lambda (\theta,v).\,\{(\theta',(e_1,\ldots,e_n))\},\,X))
\end{array}\]
The kMaps are combined and are pulled outwards from the query using the
following rewrite rules:
\[\begin{array}{l}
\mathrm{sMap2}(g,\mathrm{kMap}(f,X)) \rightarrow
  \mathrm{sMap2}(\lambda (k,v).\,g(k,f(v)),\,X)\\[1ex]
\mathrm{sMap1}(g,\mathrm{kMap}(f,X)) \rightarrow
  \mathrm{sMap1}(\lambda (k,v).\,g(k,f(v)),\,X)\\[1ex]
\mathrm{groupBy}(\mathrm{kMap}(f,X)) \rightarrow
  \mathrm{kMap}(\lambda s.\,\mathrm{map}(f,s),\,\mathrm{groupBy}(X))\\[1ex]
\mathrm{coGroup}(\mathrm{kMap}(f_x,\,X),\,\mathrm{kMap}(f_y,Y))\\
\hspace*{5ex}\rightarrow \mathrm{kMap}(\lambda(k,(xs,ys)).\,(k,(\mathrm{map}(f_x,xs),\mathrm{map}(f_y,ys))),\\
\hspace*{5ex}\skiptext{$\rightarrow \mathrm{kMap}($}\mathrm{coGroup}(X,Y))
\end{array}\]
where $\mathrm{map}(f,X)=\mathrm{cMap}(\lambda x.\,\{f(x)\},\,X)$.
Rewrite rules as these, when applied repeatedly, can pull out and
combine the non-homomorphic parts of a query, leaving a homomorphism
whose merge function can be derived from our annotation rules.

\section{An Example}\label{example}

Consider again the algebraic term $Q(S_1,S_2)$, presented at the end
of Section~\ref{mrql-algebra}. If we apply the transformations in
Equations~(\ref{norm2a}) through~(\ref{norm8}), we derive the following
term term $\trq{Q(S_1,S_2)}$, which propagates the join and group-by
keys to the query output:
\begin{align}
&\mathrm{sMap1}(\lambda(k,s).\,\{(k,avg(s))\},\nonumber\\
&\skiptext{$\mathrm{sM}($}\mathrm{groupBy}(\mathrm{swap}(\mathrm{sMap2}(\lambda(j,(xs,ys)).\,g(xs,ys),\,\nonumber\\
&\skiptext{$\mathrm{sM}(\mathrm{xx}($}\mathrm{mix}(\mathrm{coGroup}(\mathrm{sMap3}(\lambda x.\,\{(x.B,x)\},\,S_1),\label{tterm}\\
&\skiptext{$\mathrm{sM}(\mathrm{xx}(\mathrm{mix}(\mathrm{coGroup}($}\mathrm{sMap3}(\lambda y.\,\{(y.C,y)\},\,S_2)))))))\nonumber
\end{align}
where $avg(s)=\mathrm{reduce}(+,s)/\mathrm{reduce}(+,\mathrm{cMap}(\lambda v.\,\{1\},s))$\linebreak
and $g(xs,ys)=\mathrm{cMap}(\lambda x.\,\mathrm{cMap}(\lambda y.\,\{(x.A,y.D)\},\,ys),\,xs)$.
If we expand sMap1, sMap2, swap, and mix, then the transformed query is:
\[\begin{array}{l}
\{\,(\theta,(k,avg(s)))\\
\,|\,(\theta,(k,s))\in\mathrm{groupBy}(\{\,((k,(j,((),()))),v)\\
\skiptext{$\,|\,(\theta,(k,s))\in\mathrm{groupBy}($}\,|\,(j,(s_1,s_2))\in join,\\
\skiptext{$\,|\,(\theta,(k,s))\in\mathrm{groupBy}(\,|\,$}((),xs)\in\mathrm{groupBy}(s_1),\\
\skiptext{$\,|\,(\theta,(k,s))\in\mathrm{groupBy}(\,|\,$}((),ys)\in\mathrm{groupBy}(s_2),\\
\skiptext{$\,|\,(\theta,(k,s))\in\mathrm{groupBy}(\,|\,$}(k,v)\in g(xs,ys)\,\})\,\}
\end{array}\]
where $join$ is
\[\mathrm{coGroup}(\{\,(x.B,((),x))\,|\,x\in S_1\,\},\,\{\,(y.C,((),y))\,|\,y\in S_2\,\})\]
The output lineage $\theta$ is $(k,(j,((),())))$, where $k$ and $j$ are groupBy and coGroup keys.
But, $\mathrm{groupBy}(s_1)=\{((),\pi_2(s_1))\}$ and $\mathrm{groupBy}(s_2)=\{((),\pi_2(s_2))\}$ since
the key is $()$. Consequently, the transformed query becomes:
\[\begin{array}{l}
\{\,(\theta,(k,avg(s)))\\
\,|\,(\theta,(k,s))\in\mathrm{groupBy}(\{\,((k,(j,((),()))),v)\\
\skiptext{$\,|\,(\theta,(k,s))\in\mathrm{groupBy}($}\,|\,(j,(s_1,s_2))\in join,\\
\skiptext{$\,|\,(\theta,(k,s))\in\mathrm{groupBy}(\,|\,$}(k,v)\in g(\pi_2(s_1),\pi_2(s_2))\,\})\,\}
\end{array}\]

We now check if the transformed query~(\ref{tterm}) is a homomorphism.
Both coGroup inputs in~(\ref{tterm}) are sMap2 terms, annotated with $\uplus$,
which means that, based on Equation~(\ref{cogroup-monoid}), the coGroup is annotated with $\gb_{\Box\times\uplus}$.
From Equation~(\ref{smap2-rule}), the sMap2 operation is annotated with $\uplus$ as long as
$g(xs,ys)$ is annotated with $\uplus$.
Hence, the groupBy operation is annotated with $\gb_\uplus$, based on Equation~(\ref{groupby-monoid}).
Unfortunately, based on Equation~(\ref{smap1}), the sMap1 term is not a homomorphism as is,
because $avg(s)$ is not a homomorphism over $s$.
Based on the methods described in Section~\ref{factoring} that factor out non-homomorhic parts from terms,
the sMap1 term is broken into two terms $a(h(S_1,S_2))$, where $h(S_1,S_2)$ is:
\[\begin{array}{l}
\mathrm{sMap1}(\lambda(k,s).\,\{(k,(\mathrm{reduce}(+,s),\\
\skiptext{$\mathrm{sMap1}(\lambda(k,s).\,\{(k,($}\mathrm{reduce}(+,\mathrm{cMap}(\lambda v.\,1,s))))\},\ldots)
\end{array}\]
This is equivalent to the homomorphism $h(S_1,S_2)$ given in Section~\ref{approach}.

In addition, from Definition~\ref{query-answer}, the answer function $a(x)$
is $\pi_2(\mathrm{reduce}(\gb_\otimes,T(x)))$. Therefore,
the answer function $a(x)$, combined with the non-homomorphic part of the query is:
\[\begin{array}{l}
\mathrm{cMap}(\lambda(k,(s,c)).\,\{(k,s/c)\},\\
\skiptext{$\mathrm{cMap}($}\mathrm{reduce}(\gb_{\Box\times ((+)\times (+))},\,T(x)))
\end{array}\]
This is equivalent to the answer query given in Section~\ref{approach}.
Finally, $h(S_1,S_2)$ is a homomorphism annotated with $\gb_{\Box\times ((+)\times (+))}$,
because the reduce terms are annotated with $(+)$ (from Equation~(\ref{reduce}))
and we have a pair of homomorphisms. 
This is equivalent to the merge function $X\otimes Y$ given in Section~\ref{approach},
implemented as a partitioned join combined with aggregation.

\section{Handling Deletions}\label{deletion}

Our framework can be easily extended to handle deletion of existing data
from an input stream, in addition to insertion of new data. To handle
both insertions and deletions, a data stream $S_i$ in our framework
consists of an initial data set, followed by a continuous stream of
incremental batches $\Delta S_i$ and a continuous stream of
decremental batches $\Delta S_i'$ that arrive at regular time
intervals $\Delta t$. Then, the stream data at time $t+\Delta t$ is
$S_i\uplus\Delta S_i-\Delta S_i'$, where $X-Y$ is bag difference,
which satisfies $(X\uplus Y)-Y=X$. (An element appears in $X-Y$ as
many times as it appears in $X$, minus the number of times it appears
in $Y$.)  This dual stream of updates, can be implemented as a single
stream that contains updates tagged with + or -, to indicate insertion
or deletion. Alternatively, a stream data source may monitor two
separate directories, one for insertions and another for deletions, so
that if a new file is created, it will be treated as a new batch of
insertions or deletions, depending on the directory. In our
framework, we require that $\Delta S_i'\subseteq S_i$, that is, we can
only delete values that have already appeared in the stream in larger or equal
multiplicities. This restriction implies that there is a bag $W_i$ such that
$S_i=W_i\uplus\Delta S_i'$. Without this restriction, it would be hard
to diminish the query results on $S_i$ by the results on $\Delta S_i$
to calculate the new results.

\begin{definition}[Diminisher] The diminisher $\diff{\otimes}$ of
a monoid $\otimes$ is defined as follows:
\[\begin{array}{rcllrcllrcl}
\diff{\uplus} & = & -& \hspace*{1ex} &
\diff{+} & = & -& \hspace*{1ex} &
\diff{\Box} & = & \Box\\
\diff{\otimes\times\oplus} & = & \diff{\otimes}\times\diff{\oplus}&&
\diff{\gb_\oplus} & = & \Downarrow_{\diff{\oplus}}
\end{array}\]
where $\Downarrow_{\otimes}$ is a left-outer join defined as follows:
\[\begin{array}{rcl}
X\,\Downarrow_\otimes\,Y & \defined & \{\,(k,a\otimes b)\,|\,(k,a)\in X,\,(k,b)\in Y,\,a\not=b\,\}\\
&&\uplus\;\{\,(k,a)\,|\,(k,a)\in X,\,k\not\in\pi_1(Y)\,\}
\end{array}\]
\end{definition}
\noindent
Note that a diminisher is not a monoid.
\begin{theorem} For all $X,Y:$ $(X\otimes Y)\,\diff{\otimes}\,Y=X$.
\end{theorem}
\noindent
The proof is given in \short{the extended version
of this paper~\cite{extended-stream}}\extended{Appendix~\ref{proofs}}.
For $X=\otimes_z$, it implies that $Y\,\diff{\otimes}\,Y=\otimes_z$.
\begin{theorem}
If $[S_1\!:\uplus,\ldots,S_n\!:\uplus]\vdash f(\overline{S}):\otimes$\hfill
and\hfill for\hfill all\hfill $i:$\\[1ex]
$\Delta S_i'\subseteq S_i$,
then $f(\overline{S-\Delta S'})=f(\overline{S})\,\diff{\otimes}\,f(\overline{\Delta S'})$.
\end{theorem}
\begin{proof}
Since $\Delta S_i'\subseteq S_i$, then there must be $W_i$ such that $S_i=W_i\uplus\Delta S_i'$.
Thus, $S_i-\Delta S_i'=W_i\uplus\Delta S_i'-\Delta S_i'\;\Rightarrow\;W_i=S_i-\Delta S_i'$. Then,
\begin{align*}
&f(\overline{S})\,\diff{\otimes}\,f(\overline{\Delta S'})
=f(\overline{W\uplus\Delta S'})\,\diff{\otimes}\,f(\overline{\Delta S'})\\
&=(f(\overline{W})\otimes f(\overline{\Delta S'}))\,\diff{\otimes}\,f(\overline{\Delta S'})
=f(\overline{W})=f(\overline{S-\Delta S'})&&\qedhere
\end{align*}
\end{proof}
\noindent
This theorem indicates that, to process the deletions $\overline{\Delta S'}$,
we can simply use the same methods as for insertions, but using the diminishing
function $\diff{\otimes}$ for merging the state, instead of $\otimes$.
That is, we can use the same functions $h$ and $a$, but now
the state is diminished as follows:
\[\begin{array}{l}
\mathrm{state}\;\leftarrow\;\mathrm{state}\;\diff{\otimes}\; h(\overline{\Delta S'})\\
\mathbf{return}\; a(\mathrm{state})
\end{array}\]
For example, the join-groupby query
used in Section~\ref{approach} must now use the merging,
$X\Downarrow_{(-)\times (-)} Y$, equal to:
\begin{lstlisting}
select (#$\theta$#, (sx-sy,cx-cy))
  from (#$\theta$#,(sx,cx)) in #$X$#,
       (#$\theta$#,(sy,cy)) in #$Y$#
union select x from x in #$X$# where #$\pi_1$#(x) not in #$\pi_1(Y)$#
\end{lstlisting}

\begin{figure*}
\scalebox{0.80}{\includegraphics{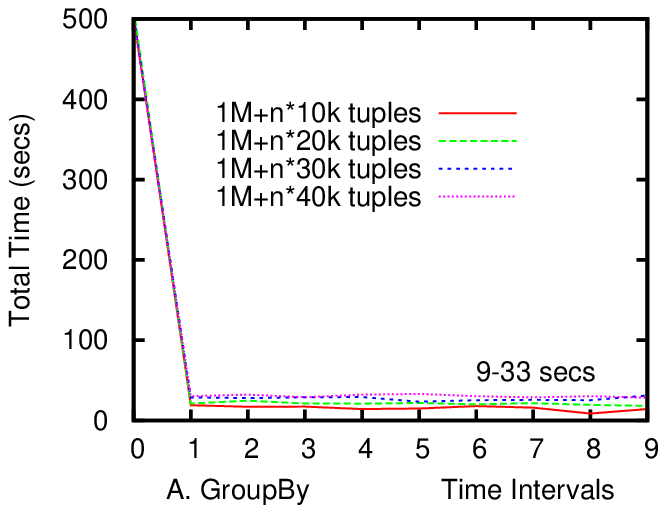}}
\scalebox{0.80}{\includegraphics{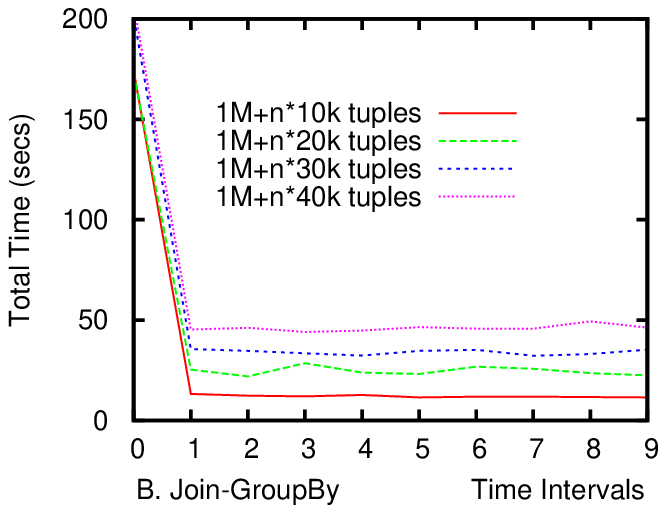}}
\scalebox{0.80}{\includegraphics{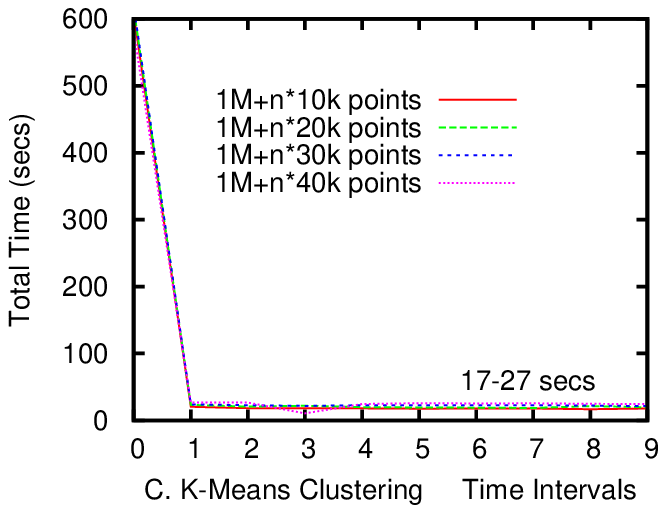}}
\caption{Incremental Query Evaluation of GroupBy, Join-GroupBy, and K-Means Clustering}\label{results}
\end{figure*}

\section{Implementation}\label{implementation}

We have implemented our incremental processing framework using Apache
MRQL~\cite{mrql} on top of Apache Spark Streaming~\cite{spark:sosp13}.
The Spark streaming engine monitors the file directories used as
stream sources in an MRQL query, and when a new file is inserted in
one of these directories or the modification time of a file changes,
it triggers the MRQL query processor to process the new files, based on
the state derived from the previous step, and creates a new state.

We have introduced a new physical operator, called Incr, which is a
stateful operator that updates a state. More
specifically, every instance of this operation is annotated with a
state number $i$ and is associated with a state, $\mathrm{state}_i$,
of type $\alpha_i$. The operation $\mathrm{Incr}(i,s_0,F)$, where $F$
is a state transition function of type $\alpha_i\rightarrow\alpha_i$
and $s_0$ is the initial state of type $\alpha_i$, has the following
semantics:
\[\begin{array}{l}
\mathrm{state}_i\;\leftarrow\;F(\mathrm{state}_i)\\
\mathbf{return}\;\mathrm{state}_i
\end{array}\]
with $\mathrm{state}_i=s_0$, initially. Note that the states
are preserved across the Inc calls and are modified by these calls.

Recall that, in our framework, we break a query $q$ into a
homomorphism $h$ and an answer function $a$, and we evaluate the
query incrementally using:
\[\begin{array}{l}
\mathrm{state}_1\;\leftarrow\;\mathrm{state}_1\,\otimes\, h(\overline{\Delta S})\\
\mathbf{return}\; a(\mathrm{state}_1)
\end{array}\]
at every time interval, where, initially, $\mathrm{state}_1=\otimes_z$.
This is implemented using the following physical plan:
\[a(\mathrm{Incr}(1,\otimes_z,\lambda T.\;T\otimes h(\overline{\Delta S})))\]
For an iteration 
$\mathrm{repeat}(\lambda X.\,f(X,\overline{\Delta S}),n,g(\overline{S}))$,
we use the approximate solution, described in Section~\ref{approach}:
we split $f$ into a homomorphism $h$, for some monoid $\otimes$, and
a function $a$, and we derive a diffusion operator $\diffuse{\otimes}$:
\[\begin{array}{l}
a(\mathrm{Incr}(1,\otimes_z,\\
\skiptext{$a(\mathrm{Incr}($}\lambda T.\;T\,\otimes\,\mathrm{repeat}(\lambda\Delta T.\;T\,\diffuse{\otimes}\, h(a(\Delta T),\overline{\Delta S}),\\
\skiptext{$a(\mathrm{Incr}(\lambda T.\;T\,\otimes\,\mathrm{repeat}($}n-1,\\
\skiptext{$a(\mathrm{Incr}(\lambda T.\;T\,\otimes\,\mathrm{repeat}($}
T\,\diffuse{\otimes}\, h(g(\overline{\Delta S}),\overline{\Delta S}))))
\end{array}\]

\ignore{
and we split $g$ into a homomorphism $h'$ for some monoid $\oplus$ and
a function $a'$.
Then, the approximate solution is implemented using the following physical plan:
\[\begin{array}{l}
a(\mathrm{Incr}(1,\otimes_z,\\
\skiptext{$a(\mathrm{Incr}($}\lambda s.\,\mathrm{repeat}(\lambda s'.\,s\otimes h(a(s'),\overline{\Delta S}),\\
\skiptext{$a(\mathrm{Incr}(\lambda s.\,\mathrm{repeat}($}n-1,\\
\skiptext{$a(\mathrm{Incr}(\lambda s.\,\mathrm{repeat}($}
s\otimes h(a'(\mathrm{Incr}(2,\oplus_z,\\
\skiptext{$a(\mathrm{Incr}(\lambda s.\,\mathrm{repeat}(s\otimes h(a'(\mathrm{Incr}($}
\lambda s''\!.\,s''\!\oplus h'(\overline{\Delta S}))),\,
\overline{\Delta S}))))
\end{array}\]
That is, the first iteration step is unrolled into the initial
state of repeat to handle the case where $X=a'(\mathrm{state}')$ in the first
call to $h$.
}

\section{Performance Evaluation}\label{performance}

The system described in this paper is available as part of the latest
official MRQL release  (MRQL 0.9.6). We have experimentally validated
the effectiveness of our methods using four queries: groupBy,
join-groupBy, k-means clustering, and PageRank. 
The platform used for our evaluations is a small cluster of 9 Linux servers,
connected through a Gigabit Ethernet switch. Each server has 4 Xeon
cores at 3.2GHz with 4GB memory. For our experiments, we used Hadoop
2.2.0 (Yarn) and Spark 1.6.0. The cluster frontend was used
exclusively as a NameNode/ResourceManager, while the rest 8 compute nodes
were used as DataNodes/NodeManagers. For our experiments, we used all
the available 32 cores of the compute nodes for Spark tasks.
\ignore{
\short{The scripts to reproduce the experiments are available at~\cite{extended-stream}.}
\extended{The scripts to reproduce the experiments are available at
\url{http://lambda.uta.edu/stream-results.zip}.}
}

The data streams used by the first 3 queries (groupBy,
join-groupBy, and k-means clustering) consist of a large set of
initial data, which is used to initialize the state, followed by a
sequence of 9 equal-size batches of data (the increments). The
experiments were repeated for increments of size 10K, 20K, 30K, and
40K tuples, always starting with a fresh state (constructed from the
initial data only). The performance results are shown in
Figure~\ref{results}. The $x$-axis represents the time points $\Delta
t$ when we get new batches of data in the stream. At time $0\Delta t$,
we have the processing of the initial data and the construction of the
initial state. Then, the 9 increments arrive at the time points
$1\Delta t$ through $9\Delta t$. The $y$-axis is the query execution
time, and there are 4 plots, one for each increment size.

The join-groupBy and the k-means clustering queries are given in
Section~\ref{approach}. The groupBy query is `\s{select (x,avg(y)) from (x,y) in
  $S$ group by x}'. The datasets used for both the groupBy and
join-groupBy queries consist of pairs of random integers between 0 and
10000. The groupBy initial dataset has size 1M tuples, while the two
join-groupBy inputs have size 100K tuples. The datasets used for the
k-means query consist of random $(X,Y)$ points in 4 squares that have
$X$ in $[2,4]$ or $[6,8]$ and $Y$ in $[2,4]$ or $[6,8]$. Thus, the 4
centroids are expected to be $(3,3)$, $(3,7)$, $(7,3)$, and $(7,7)$.
This also means that the state contains 4 centroids only. The initial
dataset for k-means contains 1M points and the k-means query uses 10
iteration steps. The k-means incremental program uses the approximate
solution described in Section~\ref{approach}\extended{ based on approximate floating
point equality}.

From Figure~\ref{results}, we can see that processing incremental
batches of data can give an order of magnitude speed-up compared to
processing all the data each time. Furthermore, the time to process
each increment does not substantially increase through time, despite
that the state grows with new data each time (in the case of the
groupBy and join-groupBy queries). This happens because merging states
is done with a partitioned join (implemented as a coGroup in Spark) so
that the new state created by coGroup is already partitioned by the
join key and is ready to be used for the next coGroup to handle the
next increment. Consequently, only the results of processing the new
data, which are typically smaller than the state, are shuffled across
the worker nodes before coGroup. This makes the incremental
processing time largely independent of the state size in most cases since
data shuffling is the most prevalent factor in main-memory distributed
processing systems.

\begin{figure}
\begin{center}
\scalebox{0.80}{\includegraphics{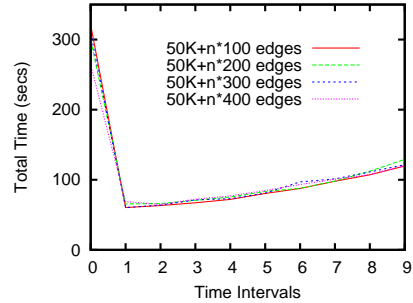}}
\caption{Incremental Query Evaluation of PageRank}\label{pagerank-results}
\end{center}
\end{figure}

We have also evaluated our system on a PageRank query on random graphs
generated by the R-MAT algorithm~\cite{Chakrabarti:sdm04} using the
Kronecker graph generator parameters: a=0.30, b=0.25, c=0.25, and
d=0.20. The PageRank query expressed in MRQL is a simple self-join on
the graph, which is optimized into a single group-by operation. (The
MRQL PageRank query is given in~\cite{edbt12}.)  The PageRanks were
calculated in 10 iteration steps.  This time, the initial graph used
in our evaluations had size 5K nodes with 50K edges and the four
different increments had sizes 100 nodes with 100, 200, 300, and 400
edges, respectively.  The performance results are shown in
Figure~\ref{pagerank-results}.  We believe that the reason that we
did not get a flat line for our incremental PageRank implementation
was due to the lack of sufficient memory tuning.  In Spark, cached RDDs
and D-Streams are not automatically garbage-collected; instead, Spark
leaves it to the programmer to cache the queued operations in memory
if their results are used multiple times, and to dispose the
cache later. This is reminiscent to legacy programming languages, such
as C, that did not support GC.  Removing all memory leaks is very crucial
to continuous stream processing, and, if not done properly, it will
eventually cause the system to run out of memory. We are planning
to fine-tune our streaming engine to remove all these memory leaks.

\section{Conclusion and Future Work}

We have presented a general framework for translating batch MRQL
queries to incremental DSPE programs. In contrast to other systems,
our methods are completely automated and are formally proven to be
correct. In addition to incremental query processing on a DSPE
platform, our framework can also be used on a batch distributed system
to process data larger than the available distributed memory, by
processing these data incrementally, in batches that can fit in
memory.  Although our methods are described using the unconventional
MRQL algebra, instead of the nested relational algebra, we believe
that many other similar query systems can use our framework by simply
translating their algebraic operators to the MRQL operators, then,
using our framework as is, and, finally, translating the resulting
operations back to their own algebra.

As a future work, we are planning to use an in-memory key-value store to
update the state in place, instead of replacing the entire state with
a new one every time we process a new batch of data. Spark
Streaming, in fact, already supports an operation updateStateByKey
operation, which allows to maintain an arbitrary state while
continuously updating it with new data. We believe that this will
considerably improve the performance of queries that require large
states, such as PageRank. Finally, in the near future, we are planning to add support
for more DSPE platforms in Incremental MRQL, such as Storm and Flink Streaming.\\[1ex]
{\bf Acknowledgments:}
This work is supported in part by the National Science Foundation 
under the grant CCF-1117369.
Some of our performance evaluations were performed at the Chameleon cloud
computing infrastructure, supported by NSF, \url{https://www.chameleoncloud.org/}.

\bibliographystyle{abbrv}

\begin{thebibliography}{10}

\bibitem{aurora}
D. J. Abadi, D. Carney, U. Cetintemel, et al.
\newblock {Aurora: A New Model and Architecture for Data Stream Management.}
\newblock In {\em VLDB Journal}, 12(2):120--139, 2003.

\extended{
\bibitem{acar:toplas09}
U. A. Acar, G. E. Blelloch, M. Blume, R. Harper, and K. Tangwongsan.
\newblock {An Experimental Analysis of Self-Adjusting Computation.}
\newblock In {\em ACM Trans. Prog. Lang. Sys., 32(1):3:1–53, 2009.}

\bibitem{acar:ddfp13}
U. A. Acar and Y. Chen.
\newblock {Streaming Big Data with Self-Adjusting Computation.}
\newblock In {\em Workshop on Data Driven Functional Programming (DDFP)}, 2013.
}

\bibitem{babcock:pods02}
B. Babcock, S. Babu, M. Datar, R. Motwani, and J. Widom.
\newblock {Models and Issues in Data Stream Systems}.
\newblock In {\em Symposium on Principles of Database Systems (PODS)}, pages 1--16, 2002.

\bibitem{benjelloun:vldb06}
O. Benjelloun, A. D. Sarma, A. Halevy, and J. Widom.
\newblock {ULDBs: Databases with Uncertainty and Lineage}.
\newblock In {\em International Conference on Very Large Data Bases (VLDB)}, pages 953--964, 2006.

\bibitem{bhagwat:vldb04}
D. Bhagwat, L. Chiticariu, W. C. Tan, and G. Vijayvargiya.
\newblock {An Annotation Management System for Relational Databases}.
\newblock In {\em International Conference on Very Large Data Bases (VLDB)}, pages 900--911, 2004.

\bibitem{incoop}
P.~Bhatotia, A.~Wieder, R.~Rodrigues, U.~A. Acar, and R.~Pasquin.
\newblock {Incoop: Mapreduce for Incremental Computations}.
\newblock  In {\em ACM Symposium on Cloud Computing (SoCC)}, 2011.

\extended{
\bibitem{cai:pldi14}
Y. Cai, P. G. Giarrusso, T. Rendel, and K. Ostermann.
\newblock {A Theory of Changes for Higher-Order Languages.
Incrementalizing $\lambda$-Calculi by Static Differentiation}.
\newblock  In {\em ACM SIGPLAN Conference on Programming Language Design and Implementation
(PLDI)}, pages 145-155, 2014.
}

\bibitem{Chakrabarti:sdm04}
D. Chakrabarti, Y. Zhan, and C. Faloutsos.
\newblock {R-MAT: A Recursive Model for Graph Mining}. 
\newblock In {\em Fourth SIAM International Conference on Data Mining (SDM)}, pages 442--446, 2004.

\bibitem{trill}
B. Chandramouli, J. Goldstein, M. Barnett, R.~DeLine, D.~Fisher, J.~C.~Platt, J.~F.~Terwilliger, J.~Wernsing.
\newblock {Trill: A High-Performance Incremental Query Processor for Diverse Analytics}.
\newblock In {\em International Conference on Very Large Data Bases (VLDB)}, pages 401--412, 2014.

\bibitem{telegraphCQ}
S.~Chandrasekaran, O.~Cooper, A.~Deshpande, M.~J.~Franklin, J.~M.~Hellerstein, W.~Hong, S.~Krishnamurthy, S.~Madden, V.~Raman, F.~Reiss, and M.~Shah.
\newblock {TelegraphCQ: Continuous Data flow Processing for an Uncertain World}.
\newblock In {\em Conference on Innovative Data System Research (CIDR)}, 2003.

\bibitem{MRonline}
T.~Condie, N.~Conway, P.~Alvaro, J.~M. Hellerstein, K.~Elmeleegy, and R.~Sears.
\newblock {Mapreduce Online}.
\newblock In {\em USENIX Symposium on Networked Systems Design and Implementation (NSDI)}, 10(4), 2010.

\bibitem{cui:vldb01}
Y. Cui and J. Widom.
\newblock {Lineage Tracing for General Data Warehouse Transformations}.
\newblock In {\em International Conference on Very Large Data Bases (VLDB)}, pages 471--480, 2001.

\bibitem{dean:osdi04}
J. Dean and S. Ghemawat.
\newblock {MapReduce: Simplified Data Processing on Large Clusters}.
\newblock In {\em Symposium on Operating System Design and Implementation (OSDI)}, 2004.

\bibitem{desikan:www05}
P. Desikan, N. Pathak, J. Srivastava, and V. Kumar
\newblock {Incremental PageRank Computation on Evolving Graphs}.
\newblock In {\em International conference on World Wide Web (WWW)},
pages 1094--1095, 2005.

\short{
\bibitem{extended-stream}
L. Fegaras.
\newblock {Incremental Query Processing on Big Data Streams (extended version)}.
\newblock \url{https://arxiv.org/abs/1511.07846}. March 2016.
}

\ignore{
\bibitem{icde08}
L. Fegaras.
\newblock {Efficient Processing of XML Update Streams}.
\newblock In {\em International Conference on Data Engineering (ICDE)}, pages 616--625, 2008.

\bibitem{dexa11}
L. Fegaras.
\newblock {Incremental Maintenance of Materialized XML Views}.
\newblock In {\em International Conference on Database and Expert Systems Applications (DEXA)}, pages 17--32, 2011.
}

\extended{
\bibitem{datacloud12}
L. Fegaras.
\newblock {Supporting Bulk Synchronous Parallelism in Map-Reduce Queries}.
\newblock In {\em International Workshop on Data Intensive Computing in the Clouds (DataCloud)}, 2012.
}

\bibitem{webdb11}
L. Fegaras, C. Li, U. Gupta, and J. J. Philip. 
\newblock {XML Query Optimization in Map-Reduce}.
\newblock In {\em International Workshop on the Web and Databases (WebDB)}, 2011.

\bibitem{edbt12}
L. Fegaras, C. Li, and U. Gupta. 
\newblock {An Optimization Framework for Map-Reduce Queries}.
\newblock In {\em International Conference on Extending Database Technology (EDBT)}, pages 26--37, 2012.

\bibitem{tods00}
L. Fegaras and D. Maier.
\newblock {Optimizing Object Queries Using an Effective Calculus}.
\newblock In {\em ACM Transactions on Database Systems (TODS)}, 25(4):457--516, 2000.

\balance

\bibitem{flink}
Apache Flink.
\url{http://flink.apache.org/}.

\bibitem{giraph}
Apache Giraph. \url{http://giraph.apache.org/}.

\extended{
\bibitem{systemML}
A. Ghoting, R. Krishnamurthy, E. Pednault, B. Reinwald, V. Sindhwani, S. Tatikonda, Y. Tian, and S. Vaithyanathan.
\newblock {SystemML: Declarative Machine Learning on MapReduce}.
\newblock In {\em IEEE International Conference on Data Engineering (ICDE)}, pages 231--242, 2011.

\bibitem{gupta:bde95}
A. Gupta and I. S. Mumick.
\newblock {Maintenance of Materialized Views: Problems, Techniques, and Applications}.
\newblock In {\em IEEE Bulletin on Data Engineering}, 18(2):145--157, 1995.
}

\bibitem{hadoop}
Apache Hadoop. \url{http://hadoop.apache.org/}.

\bibitem{hama}
Apache Hama. \url{http://hama.apache.org/}.

\bibitem{hive}
Apache Hive. \url{http://hive.apache.org/}.

\bibitem{logothetis:socc12}
D. Logothetis, C. Olston, B. Reed, K.C. Webb, and K. Yocum.
\newblock {Stateful Bulk Processing for Incremental Analytics}.
\newblock In {\em ACM Symposium on Cloud Computing (SoCC)}, 2010.

\bibitem{graphlab:vldb12}
Y. Low, J. Gonzalez, A. Kyrola, D. Bickson, C. Guestrin, and J. M. Hellerstein.
\newblock {Distributed GraphLab: A Framework for Machine Learning and Data Mining in the Cloud}.
\newblock In {\em Proceedings of the VLDB Endowment}, 5(8):716-727, 2012.

\bibitem{pregel:podc09}
G. Malewicz, M. H. Austern,  A. J.C Bik, J. C. Dehnert, I. Horn, N. Leiser, and G. Czajkowski.
\newblock {Pregel: a System for Large-Scale Graph Processing}.
\newblock In {\em ACM symposium on Principles of Distributed Computing (PODC)}, 2009.

\bibitem{mcsherry:cidr13}
F. McSherry, D. G. Murray, R. Isaacs, and M. Isard.
\newblock {Differential Dataflow}.
\newblock In {\em Conference on Innovative Data System Research (CIDR)}, 2013.

\bibitem{mihaylov:vldb12}
S. R. Mihaylov, Z. G. Ives, and S. Guha.
\newblock {REX: Recursive, Delta-Based Data-Centric Computation}.
\newblock In {\em Proceedings of the VLDB Endowment}, 5(11):1280-1291, 2012.

\bibitem{naiad}
D.~G. Murray, F.~McSherry, R.~Isaacs, M.~Isard, P.~Barham, and M.~Abadi.
\newblock {Naiad: a Timely Dataflow System}.
\newblock In {\em ACM Symposium on Operating Systems Principles (SOSP)}, 2013.

\bibitem{mrql}
Apache MRQL (incubating).
\url{http://mrql.incubator.apache.org/}.

\bibitem{olston:sigmod08}
C. Olston, B. Reed, U. Srivastava, R. Kumar, and A. Tomkins. 
\newblock {Pig Latin: a not-so-Foreign Language for Data Processing}.
\newblock In {\em ACM SIGMOD International Conference on Management of Data}, pages 1099-1110, 2008.

\bibitem{peng:osdi10}
D.~Peng and F.~Dabek.
\newblock {Large-scale Incremental Processing using Distributed Transactions and Notifications}.
\newblock In {\em Symposium on Operating System Design and Implementation (OSDI)}, 2010.

\bibitem{piccolo:osdi10}
R. Power and J. Li.
\newblock {Piccolo: Building Fast, Distributed Programs with Partitioned Tables}.
\newblock In {\em Symposium on Operating System Design and Implementation (OSDI)}, 2010.

\extended{
\bibitem{ramalingam:popl93}
Ramalingam, G. and Reps, T. 
\newblock {A Categorized Bibliography on Incremental Computation}.
\newblock In {\em Principles of Programming Languages (POPL)}, 1993, pp. 502–510.
}

\bibitem{shinnar:vldb12}
A. Shinnar, D. Cunningham, B. Herta, and V. Saraswat.
\newblock {M3R: Increased Performance for In-Memory Hadoop Jobs}.
\newblock In {\em Proceedings of the VLDB Endowment}, 5(12):1736-1747, 2012.

\bibitem{samza}
Apache Samza.
\url{http://samza.apache.org/}

\bibitem{spark}
Apache Spark.
\url{http://spark.apache.org/}.

\bibitem{storm}
\newblock {Apache Storm: A System for Processing Streaming Data in Real Time}.
\url{http://hortonworks.com/hadoop/storm/}.

\bibitem{S4}
Apache S4 (incubating): A Distributed Stream Computing Platform.
\url{http://incubator.apache.org/s4/}.

\bibitem{tangwongsan:vldb15}
K. Tangwongsan, M. Hirzel, S. Schneider, and K.-L. Wu.
\newblock {General Incremental Sliding-Window Aggregation}.
\newblock In {\em Proceedings of the VLDB Endowment}, 8(7):702-713, 2015.

\bibitem{valiant:cacm90}
L. G. Valiant.
\newblock {A Bridging Model for Parallel Computation}.
\newblock In {\em Communications of the ACM (CACM)}, 33(8):103-111, August 1990.

\bibitem{spark:nsdi12}
M. Zaharia, M. Chowdhury, T. Das, A. Dave, J. Ma, M. McCauley, M. J. Franklin, S. Shenker, and I. Stoica.
\newblock {Resilient Distributed Datasets: A Fault-Tolerant Abstraction for In-Memory Cluster Computing}.
\newblock In {\em USENIX Symposium on Networked Systems Design and Implementation (NSDI)}, 2012.

\bibitem{spark:sosp13}
M. Zaharia, T. Das, H. Li, T. Hunter, S. Shenker, and I. Stoica.
\newblock {Discretized Streams: Fault-Tolerant Streaming Computation at Scale}.
\newblock In {\em Symposium on Operating Systems Principles (SOSP)}, 2013.

\bibitem{i2mapreduce:tkde15}
Y. Zhang, S. Chen, Q. Wang, and G. Yu.
\newblock {i2 MapReduce: Incremental MapReduce for Mining Evolving Big Data}. 
\newblock In {\em IEEE Transactions on Knowledge and Data Engineering (TKDE)}, 27(7):1906-1919, 2015.

\end{thebibliography}

\extended{
\begin{appendix}
\section{Proofs}\label{proofs}

\begin{proof}[Proof of Equation~(\ref{groupby-unnest})]
We will use structural induction. The inductive step
\[\{\,(k,v)\,|\,(k,s)\in\mathrm{groupBy}(X\uplus Y),\,v\in s\,\}\; =\; X\uplus Y\]
can be proven as follows.
Let $GX=\mathrm{groupBy}(X)$ and $GY=\mathrm{groupBy}(Y)$. Then, for
\begin{align}
M(X,Y)\; =\; &\{\,(k,a)\,|\,(k,a)\in X,\,(k,b)\in Y\,\}\label{m-eq}\\
&\uplus\{\,(k,a)\,|\,(k,a)\in X,\,k\not\in \pi_1(Y)\,\}\nonumber
\end{align}
we have, $M(GX,GY)=GX$, because
for each group-by key $k$, there is at most one $(k,a)\in GX$ and at most one $(k,b)\in GY$.
Then, we have:
\begin{align*}
&\{\,(k,v)\,|\,(k,s)\in\mathrm{groupBy}(X\uplus Y),\,v\in s\,\}\\
&=\{\,(k,v)\,|\,(k,s)\in(GX\gb_\uplus GY),\,v\in s\,\}\\
&=\{\,(k,v)\,|\,(k,s)\in(\{\,(k,a\uplus b)\,|\,(k,a)\in GX,\,(k,b)\in GY\,\}\\
&\skiptext{$=\{\,(k,v)\,|\,$}\uplus\;\{\,(k,a)\,|\,(k,a)\in GX,\,k\not\in\pi_1(GY)\,\}\\
&\skiptext{$=\{\,(k,v)\,|\,$}\uplus\;\{\,(k,b)\,|\,(k,b)\in GY,\,k\not\in\pi_1(GX)\,\}),\,v\in s\,\}\\
&=\{\,(k,v)\,|\,(k,a)\in GX,\,(k,b)\in GY,\,v\in(a\uplus b)\,\}\\
&\skiptext{$=$}\uplus\,\{\,(k,v)\,|\,(k,a)\in GX,\,k\not\in\pi_1(GY),v\in a\,\}\\
&\skiptext{$=$}\uplus\,\{\,(k,v)\,|\,(k,b)\in GY,\,k\not\in\pi_1(GX),v\in b\,\}\\
&=\{\,(k,v)\,|\,(k,a)\in GX,\,(k,b)\in GY,\,v\in a\,\}\\
&\skiptext{$=$}\uplus\,\{\,(k,v)\,|\,(k,a)\in GX,\,(k,b)\in GY,\,v\in b\,\}\\
&\skiptext{$=$}\uplus\,\{\,(k,v)\,|\,(k,a)\in GX,\,k\not\in\pi_1(GY),v\in a\,\}\\
&\skiptext{$=$}\uplus\,\{\,(k,v)\,|\,(k,b)\in GY,\,k\not\in\pi_1(GX),v\in b\,\}\\
&=\{\,(k,v)\,|\,(k,a)\in M(GX,GY),\,v\in a\,\}\\
&\skiptext{$=$}\uplus\,\{\,(k,v)\,|\,(k,b)\in M(GY,GX),\,v\in b\,\}\\
&=\{\,(k,v)\,|\,(k,a)\in GX,\,v\in a\,\}\\
&\skiptext{$=$}\uplus\,\{\,(k,v)\,|\,(k,b)\in GY,\,v\in b\,\}\\
&= X\uplus Y\hspace*{5ex}\mbox{\em (from induction hypothesis)}&&\qedhere
\end{align*}
\end{proof}

\begin{proof}[Proof of Lemma~\ref{smap-hom} (Transformed Term Judgments)]
Using\linebreak Equations~(\ref{smap1}) and~(\ref{gb-def}),
Judgement~(\ref{smap1-rule}) can be proven as follows:
\begin{align*}
&\mathrm{sMap1}(f,X\gb_\oplus Y)\\
&=\{\,((\theta,k),b)\,|\,((\theta,k),a)\in(X\gb_\oplus Y),\,b\in f(k,a)\,\}\\
&=\{\,((\theta,k),b)\,|\,((\theta,k),x)\in X,\,((\theta,k),y)\in Y,\,b\in f(k,x\oplus y)\,\}\,\uplus\cdots\\
&=\{\,((\theta,k),b)\,|\,((\theta,k),x)\in X,\,((\theta,k),y)\in Y,\\
&\skiptext{$=\{\,((\theta,k),b)\,|\,$}b\in(f(k,x)\gb_\otimes f(k,y))\,\}\,\uplus\cdots\\
&=\{\,((\theta,k),n\otimes m)\,|\,((\theta,k),x)\in X,\,((\theta,k),y)\in Y,\\
&\skiptext{$=\{\,((\theta,k),n\otimes m)\,|\;$}n\in f(k,x),\,m\in f(k,y)\,\}\,\uplus\cdots\\
&=\{\,((\theta,k),n\otimes m)\,|\,((\theta,k),x)\in X,\,n\in f(k,x),\\
&\skiptext{$=\{\,((\theta,k),n\otimes m)\,|\;$}((\theta,k),y)\in Y,\,m\in f(k,y)\,\}\,\uplus\cdots\\
&=\{\,((\theta,k),n\otimes m)\,|\,((\theta,k),n)\in\mathrm{sMap1}(f,X),\\
&\skiptext{$=\{\,((\theta,k),n\otimes m)\,|\;$}((\theta,k),m)\in\mathrm{sMap1}(f,Y)\,\}\,\uplus\cdots\\
&=\mathrm{sMap1}(f,X)\gb_\otimes\mathrm{sMap1}(f,Y)
\end{align*}
Using Equations~(\ref{smap2}),~(\ref{gb-def}), and~(\ref{swap}),
Judgement~(\ref{smap2-rule}) can be proven as follows:
\[\begin{array}{l}
\mathrm{sMap2}(f,X\gb_\oplus Y)\\
=\mathrm{sMap2}(f,\{\,((k,\theta),a\oplus b)\,|\,((k,\theta),a)\in X,\,((k,\theta),b)\in Y\,\}\\
\skiptext{$==$}\uplus\cdots)\\
=\{\,(k',((k,\theta),d))\\
\skiptext{$==$}|\,((k,\theta),c)\in (\{\,((k,\theta),a\oplus b)\,|\,((k,\theta),a)\in X,\\
\skiptext{$=\;|\,((k,\theta),c)\in (\{\,((k,\theta),a\oplus b)\,|\,$}((k,\theta),b)\in Y\,\}\,\uplus\cdots),\\
\skiptext{$==|\,$}(k',d)\in f(k,c)\,\}\\
=\{\,(k',((k,\theta),d))\,|\,((k,\theta),a)\in X,\,((k,\theta),b)\in Y,\\
\skiptext{$=\{\,(k',((k,\theta),d))\,|\,$}(k',d)\in f(k,a\oplus b)\,\}\,\uplus\cdots\\
=\{\,(k',((k,\theta),d))\,|\,((k,\theta),a)\in X,\,((k,\theta),b)\in Y,\\
\skiptext{$=\{\,(k',((k,\theta),d))\,|\,$}(k',d)\in (f(k,a)\uplus f(k,b))\,\}\,\uplus\cdots\\
=\{\,(k',((k,\theta),d))\,|\,((k,\theta),a)\in X,\,((k,\theta),b)\in Y,\\
\skiptext{$=\{\,(k',((k,\theta),d))\,|\,$}(k',d)\in f(k,a)\,\}\\
\skiptext{$=$}\uplus\,\{\,(k',((k,\theta),d))\,|\,((k,\theta),a)\in X,\,((k,\theta),b)\in Y,\\
\skiptext{$=\uplus\,\{\,(k',((k,\theta),d))\,|\,$}(k',d)\in f(k,b)\,\}\,\uplus\cdots\\
=(\{\,(k',((k,\theta),d))\,|\,((k,\theta),a)\in X,\,((k,\theta),b)\in Y,\\
\skiptext{$=(\{\,(k',((k,\theta),d))\,|\,$}(k',d)\in f(k,a)\,\}\\
\skiptext{$=($}\uplus\,\{\,(k',((k,\theta),d))\,|\,((k,\theta),a)\in X,\,(k,\theta)\not\in\pi_1(Y),\\
\skiptext{$=(\uplus\,\{\,(k',((k,\theta),d))\,|\,$}(k',d)\in f(k,a)\,\})\\
\skiptext{$=$}\uplus\,(\{\,(k',((k,\theta),d))\,|\,((k,\theta),a)\in X,\,((k,\theta),b)\in Y,\\
\skiptext{$=\uplus\,(\{\,(k',((k,\theta),d))\,|\,$}(k',d)\in f(k,b)\,\}\\
\skiptext{$=\uplus\,($}\uplus\,\{\,(k',((k,\theta),d))\,|\,((k,\theta),b)\in Y,\,(k,\theta)\not\in\pi_1(X),\\
\skiptext{$=\uplus\,(\uplus\,\{\,(k',((k,\theta),d))\,|\,$}(k',d)\in f(k,b)\,\})\\
=\{\,(k',((k,\theta),d))\,|\,((k,\theta),a)\in X,\,(k',d)\in f(k,a)\,\})\\
\skiptext{$=$}\uplus\,\{\,(k',((k,\theta),d))\,|\,((k,\theta),b)\in Y,\,(k',d)\in f(k,b)\,\}\\
=\mathrm{sMap2}(f,X)\uplus\mathrm{sMap2}(f,Y)
\end{array}\]
Judgement~(\ref{smap3-rule}) is a direct consequence of Judgement~(\ref{cmap-box}).
Judgement~(\ref{mix-rule}) can be proven as follows:
\begin{align*}
&\mathrm{mix}(X\gb_{\uplus\times\uplus}Y)\\
&=\{\,((k,(\theta_x,\theta_y)),(xs,ys))\,|\,(k,(s_1,s_2))\in (X\gb_{\uplus\times\uplus}Y),\\
&\hspace*{8ex}(\theta_x,xs)\in\mathrm{groupBy}(s_1),\,(\theta_y,ys)\in\mathrm{groupBy}(s_2)\,\}\\
&=\{\,((k,(\theta_x,\theta_y)),(xs,ys))\,|\,(k,(s^1_1,s^1_2))\in X,\,(k,(s^2_1,s^2_2))\in Y,\\
&\hspace*{8ex}(\theta_x,xs)\in\mathrm{groupBy}(s^1_1\uplus s^2_1),\\
&\hspace*{8ex}(\theta_y,ys)\in\mathrm{groupBy}(s^1_2\uplus s^2_2)\,\}\,\uplus\cdots\\
&=\{\,((k,(\theta_x,\theta_y)),(xs_1\uplus xs_2,ys_1\uplus ys_2))\\
&\hspace*{4ex}|\,(k,(s^1_1,s^1_2))\in X,\,(k,(s^2_1,s^2_2))\in Y,\\
&\hspace*{5ex}(\theta_x,xs_1)\in\mathrm{groupBy}(s^1_1),\,(\theta_x,xs_2)\in\mathrm{groupBy}(s^2_1),\\
&\hspace*{5ex}(\theta_y,ys_1)\in\mathrm{groupBy}(s^1_2),\,(\theta_y,ys_2)\in\mathrm{groupBy}(s^2_2)\,\}\,\uplus\cdots\\
&=\{\,((k,(\theta_x,\theta_y)),(xs_1\uplus xs_2,ys_1\uplus ys_2))\\
&\hspace*{4ex}|\,((k,(\theta_x,\theta_y)),(xs_1,ys_1))\in\mathrm{mix}(X),\\
&\hspace*{5ex}((k,(\theta_x,\theta_y)),(xs_2,ys_2))\in\mathrm{mix}(Y)\,\}\,\uplus\cdots\\
&=\mathrm{mix}(X)\gb_{\uplus\times\uplus}\mathrm{mix}(Y)&&\qedhere
\end{align*}
\end{proof}

\setcounter{theorem}{0}
\begin{theorem}[Homomorphism]
Any transformed term $\trq{q}$, where $q$ is defined in
Definition~\ref{normalized-algebra}, is a homomorphism over the input
streams, provided that each generated sMap1 and sMap2 term in $\trq{q}$ satisfies
the premises in Judgment~(\ref{smap1-rule}) and~(\ref{smap2-rule}):
\begin{gather}
\rho[S_1:\uplus,\ldots,S_n:\uplus]\vdash\trq{q}:\monoid{q}\tag{\ref{cc2}}
\end{gather}
\end{theorem}
\begin{proof}
Using induction, we can first prove that $\trc{c}:\uplus$,\linebreak $\tre{\mathrm{groupBy}(c)}:\,\gb_\uplus$, and
$\tre{\mathrm{coGroup}(c_1,c_2)}:\,\gb_{\uplus\times\uplus}$, for all $c,c_1,c_2$.
We have $\trc{\mathrm{cMap}(f,S_i)}=\mathrm{sMap3}(f,S_i):\uplus$ from Judgment~(\ref{smap3-rule}),
and $\trc{\mathrm{cMap}(f,e)}=\mathrm{sMap2}(f,\tre{e}):\uplus$ from Judgment~(\ref{smap2-rule}),
because $\trq{e}$ is either $\gb_\uplus$ or $\gb_{\uplus\times\uplus}$.
Similarly, $\tre{\mathrm{groupBy}(c)}=\mathrm{groupBy}(\mathrm{swap}(\trc{c})):\,\gb_\uplus$
from Judgments~(\ref{groupby-monoid}) and~(\ref{swap-rule}), and $\tre{\mathrm{coGroup}(c_1,c_2)}$\linebreak
$=\mathrm{mix}(\mathrm{coGroup}(\trc{c_1},\trc{c_2})):\,\gb_{\uplus\times\uplus}$
from Judgments~(\ref{cogroup-monoid}) and~(\ref{mix-rule}).
Then, $\trq{\mathrm{cMap}(f,e)}=\mathrm{sMap1}(f,\tre{e}):\,\gb_\otimes$ from Judgment~(\ref{smap1-rule}), and
$\trq{\mathrm{reduce}(\oplus,\mathrm{cMap}(f,e))}=\mathrm{reduce}(\gb_\oplus,\mathrm{sMap1}(f,\tre{e})):\,\gb_\oplus$
from Judgments~(\ref{reduce}) and~(\ref{smap1-rule}).
\end{proof}

We will prove Theorem~\ref{correctness} using the following lemma:
\begin{lemma}\label{lemma}
The following equations apply to terms in Definition~\ref{normalized-algebra}
for a function $f$ that satisfies the premise of Equation~(\ref{smap2-rule}):
\begin{subequations}
\begin{align}
G_1(\mathrm{groupBy}(\mathrm{swap}(X))) =&\;\mathrm{groupBy}(T'(X))\label{c1}\\
G_2(\mathrm{mix}(\mathrm{coGroup}(X,Y))) =&\;\mathrm{coGroup}(T'(X),T'(Y))\label{c2}\\
T'(\mathrm{sMap2}(f,X)) = \mathrm{cMap}&(f,G_1(X)),\hspace*{1ex}\mbox{if $X\!:\,\gb_\uplus$}\label{c3}\\
T'(\mathrm{sMap2}(f,X)) = \mathrm{cMap}&(f,G_2(X)),\hspace*{1ex}\mbox{if $X\!:\,\gb_{\uplus\times\uplus}$}\label{c4}
\end{align}
\end{subequations}
where $T'$ removes the lineage $\theta$, while $G_1$ and $G_2$, in addition to removing $\theta$,
reconstruct the groupBy and coGroup result:
\begin{subequations}
\begin{align}
T'(X) = & \{\,(k,a)\,|\,(k,(\theta,a))\in X\,\}\label{TT}\\
G_1(X) = & \;\mathrm{groupBy}(\{\,(k,a)\,|\,((k,\theta),s)\in X,\,a\in s\,\})\label{G1}\\
G_2(X) = & \;\mathrm{coGroup}(\{\,(k,a)\,|\,((k,\theta),(s_1,s_2))\in X,\,a\in s_1\,\},\nonumber\\
& \skiptext{$\;\mathrm{co}($}\{\,(k,b)\,|\,((k,\theta),(s_1,s_2))\in X,\,b\in s_2\,\})\label{G2}
\end{align}
\end{subequations}
\end{lemma}
\begin{proof}
Proof of Equation~(\ref{c1}) using Equation~(\ref{groupby-unnest}):
\[\begin{array}{l}
G_1(\mathrm{groupBy}(\mathrm{swap}(X)))\\
= \mathrm{groupBy}(\{\,(k,a)\,|\,((k,\theta),s)\in\mathrm{groupBy}(\mathrm{swap}(X)),\,a\in s\,\})\\
= \mathrm{groupBy}(T(\{\,((k,\theta),a)\,|\,((k,\theta),s)\in \mathrm{groupBy}(\mathrm{swap}(X)),\\
\skiptext{$=\mathrm{groupBy}(T'(\{\,((k,\theta),a)\,|\,$}a\in s\,\}))\\
= \mathrm{groupBy}(T(\mathrm{swap}(X)))\\
= \mathrm{groupBy}(T'(X))
\end{array}\]
Proof of Equation~(\ref{c2}):
\[\begin{array}{l}
G_2(\mathrm{mix}(\mathrm{coGroup}(X,Y)))\\
=\mathrm{coGroup}(\{\,(k,a)\,|\,((k,\theta),(s_1,s_2))\in\mathrm{mix}(\mathrm{coGroup}(X,Y)),\\
\skiptext{$=\mathrm{coGroup}(\{\,(k,a)\,|\,$}a\in s_1\,\},\\
\skiptext{$=\mathrm{coGroup}($}\{\,(k,b)\,|\,((k,\theta),(s_1,s_2))\in\mathrm{mix}(\mathrm{coGroup}(X,Y)),\\
\skiptext{$=\mathrm{coGroup}(\{\,(k,a)\,|\,$}b\in s_2\,\})\\
=\mathrm{coGroup}(\{\,(k,a)\,|\,(k,(q_1,q_2))\in\mathrm{coGroup}(X,Y),\\
\skiptext{$=\mathrm{coGroup}(\{\,(k,a)\,|\,$}(\theta_x,xs)\in\mathrm{groupBy}(q_1),\\
\skiptext{$=\mathrm{coGroup}(\{\,(k,a)\,|\,$}(\theta_y,ys)\in\mathrm{groupBy}(q_2),\,a\in xs\,\},\\
\skiptext{$=\mathrm{coGroup}($}\{\,(k,b)\,|\,(k,(q_1,q_2)))\in\mathrm{coGroup}(X,Y),\\
\skiptext{$=\mathrm{coGroup}(\{\,(k,b)\,|\,$}(\theta_x,xs)\in\mathrm{groupBy}(q_1),\\
\skiptext{$=\mathrm{coGroup}(\{\,(k,b)\,|\,$}(\theta_y,ys)\in\mathrm{groupBy}(q_2),\,b\in ys\,\})\\
=\mathrm{coGroup}(\{\,(k,a)\,|\,(k,(q_1,q_2))\in\mathrm{coGroup}(X,Y),\\
\skiptext{$=\mathrm{coGroup}(\{\,(k,a)\,|\,$}a\in\pi_2(q_1)\,\},\\
\skiptext{$=\mathrm{coGroup}($}\{\,(k,b)\,|\,(k,(q_1,q_2))\in\mathrm{coGroup}(X,Y),\\
\skiptext{$=\mathrm{coGroup}(\{\,(k,b)\,|\,$}b\in\pi_2(q_2)\,\})\\
=\mathrm{coGroup}(\{\,(k,a)\,|\,(k,(s_1,s_2))\in\mathrm{coGroup}(T'(X),T'(Y)),\\
\skiptext{$=\mathrm{coGroup}(\{\,(k,a)\,|\,$}a\in s_1\,\},\\
\skiptext{$=\mathrm{coGroup}($}\{\,(k,b)\,|\,(k,(s_1,s_2))\in\mathrm{coGroup}(T'(X),T'(Y)),\\
\skiptext{$=\mathrm{coGroup}(\{\,(k,b)\,|\,$}b\in s_2\,\})\\
=\mathrm{coGroup}(T'(X),T'(Y))
\end{array}\]
It is easy to prove that, for an $f$ that satisfies the premise of Equation~(\ref{smap2-rule}),
we have $\mathrm{cMap}(f,X):\uplus$ for $X:\,\gb_\uplus$.
Then, we can prove Equation~(\ref{c3}) for $X\gb_\uplus Y$ using Equation~(\ref{smap2-rule}) and induction:
\begin{align*}
&\mathrm{cMap}(f,G_1(X\gb_\uplus Y))\\
&=\mathrm{cMap}(f,(G_1(X)\gb_\uplus G_1(Y)))\\
&=\mathrm{cMap}(f,G_1(X))\uplus\mathrm{cMap}(f,G_1(Y))\\
&=T'(\mathrm{sMap2}(f,X))\uplus T'(\mathrm{sMap2}(f,Y))\\
&=T'(\mathrm{sMap2}(f,X\gb_\uplus Y))
\end{align*}
The proof of Equation~(\ref{c4}) is similar.
\end{proof}

\begin{theorem}[Correctness]
The $\answer{q}_x$ over the state $x=\trq{q}$ returns the same result as the original query $q$,
where $q$ is defined in Definition~\ref{normalized-algebra}:
\begin{align}
&\answer{q}_x = q & \mbox{for $x=\trq{q}$}\tag{\ref{cc1}}
\end{align}
\end{theorem}
\begin{proof}
To prove Equation~(\ref{cc1}), we first prove $T'(\trc{c})=c$ from Equations~(\ref{norm7}) and~(\ref{norm8}).
If $c=\mathrm{cMap}(f,\,S_i)$, then
\[\begin{array}{l}
T'(\trc{\mathrm{cMap}(f,\,S_i)})\\
=T'(\mathrm{sMap3}(f,\,S_i))\\
=\mathrm{cMap}(f,S_i)
\end{array}\]
If $e=\mathrm{cMap}(f,\,\mathrm{groupBy}(c))$, then using Equations~(\ref{norm8}), (\ref{norm5}), (\ref{c1}) and~(\ref{c3}):
\[\begin{array}{l}
T'(\trc{\mathrm{cMap}(f,\,\mathrm{groupBy}(c))})\\
=T'(\mathrm{sMap2}(f,\,\tre{\mathrm{groupBy}(c)}))\\
=T'(\mathrm{sMap2}(f,\,\mathrm{groupBy}(\mathrm{swap}(\trc{c}))))\\
=\mathrm{cMap}(f,\,G_1(\mathrm{groupBy}(\mathrm{swap}(\trc{c}))))\\
=\mathrm{cMap}(f,\,\mathrm{groupBy}(T'(\trc{c})))\\
=\mathrm{cMap}(f,\,\mathrm{groupBy}(c))
\end{array}\]
If $e=\mathrm{coGroup}(c_1,c_2)$, then using Equations~(\ref{norm8}), (\ref{norm6}), (\ref{c2}) and~(\ref{c4})):
\[\begin{array}{l}
T'(\trc{\mathrm{cMap}(f,\,\mathrm{coGroup}(c_1,c_2))})\\
=T'(\mathrm{sMap2}(f,\,\tre{\mathrm{coGroup}(c_1,c_2)}))\\
=T'(\mathrm{sMap2}(f,\,\mathrm{mix}(\mathrm{coGroup}(\trc{c_1},\trc{c_2}))))\\
=\mathrm{cMap}(f,\,G_2(\mathrm{mix}(\mathrm{coGroup}(\trc{c_1},\trc{c_2}))))\\
=\mathrm{cMap}(f,\,\mathrm{coGroup}(T'(\trc{c_1}),T'(\trc{c_2})))\\
=\mathrm{cMap}(f,\,\mathrm{coGroup}(c_1,c_2))
\end{array}\]

We now prove Equation~(\ref{cc1}) using induction on $q$.
Based on Definition~\ref{query-answer}, we have three cases.
When $q=(q_1,q_2)$, then for $x=\trq{(q_1,q_2)}=(\trq{q_1},\trq{q_2})$:
\begin{align*}
&\answer{(q_1,q_2)}_x\\
&=(\answer{q_1}_{\pi_1(x)},\,\answer{q_2}_{\pi_2(x)})\\
&=(q_1,\,q_2) &\mbox{(from induction hypothesis)}
\end{align*}
For $q=\mathrm{cMap}(f,\,S_i)$, we have:
\[\begin{array}{l}
\answer{\mathrm{cMap}(f,\,S_i)}_x\\
=\pi_2(x)\\
=\pi_2(\trq{\mathrm{cMap}(f,\,S_i)})\\
=\pi_2(\{\,((),b)\,|\,a\in S_i,\,b\in f(a)\,\})\\
=\mathrm{cMap}(f,\,S_i)
\end{array}\]
For $q=\mathrm{cMap}(f,\,\mathrm{groupBy}(c))$, we have:
\[\begin{array}{l}
\answer{\mathrm{cMap}(f,\,\mathrm{groupBy}(c))}_x
=\pi_2(\mathrm{reduce}(\gb_\otimes,T(x)))
\end{array}\]
where
\[\begin{array}{rcl}
x & = & \trq{\mathrm{cMap}(f,\,\mathrm{groupBy}(c))}\\
  & = & \mathrm{sMap1}(f,\,\mathrm{groupBy}(\mathrm{swap}(\trc{c})))
\end{array}\]
It is easy to prove by induction that for $X:\;\gb_\otimes$ we have
$\pi_2(\mathrm{reduce}(\gb_\otimes,T(\mathrm{sMap1}(f,X))))=\mathrm{cMap}(f,G_1(X))$.
Then, 
\[\begin{array}{l}
\pi_2(\mathrm{reduce}(\gb_\otimes,T(\mathrm{sMap1}(f,\,\mathrm{groupBy}(\mathrm{swap}(\trc{c}))))))\\
=\mathrm{cMap}(f,\,G_1(\mathrm{groupBy}(\mathrm{swap}(\trc{c}))))\\
=\mathrm{cMap}(f,\,\mathrm{groupBy}(T(\trc{c})))\\
=\mathrm{cMap}(f,\,\mathrm{groupBy}(c))
\end{array}\]
The proof for $q=\mathrm{cMap}(f,\,\mathrm{coGroup}(c_1,c_2))$ is similar.
For $q=\mathrm{reduce}(\oplus,\,\mathrm{cMap}(f,\,e))$,
and for $e=\mathrm{groupBy}(c)$ or $e=\mathrm{coGroup}(c_1,c_2)$
(the proof is easier for $e=S_i$), we have:
\begin{align*}
&\answer{\mathrm{reduce}(\oplus,\,\mathrm{cMap}(f,\,e))}_x\\
&=\mathrm{reduce}(\oplus,\,\pi_2(x))\\
&=\mathrm{reduce}(\oplus,\,\pi_2(\trq{\mathrm{reduce}(\oplus,\,\mathrm{cMap}(f,\,e))}))\\
&=\mathrm{reduce}(\oplus,\,\pi_2(\mathrm{reduce}(\gb_\oplus,\,\mathrm{sMap1}(f,\,\tre{e}))))\\
&=\mathrm{reduce}(\oplus,\,\pi_2(\mathrm{reduce}(\gb_\oplus,\,\trq{\mathrm{cMap}(f,\,e)})))\\
&=\mathrm{reduce}(\oplus,\,\mathrm{cMap}(f,\,e))
\end{align*}
from the previous proofs for $q=\mathrm{cMap}(f,\,\mathrm{groupBy}(c))$ and\linebreak
$q=\mathrm{cMap}(f,\,\mathrm{coGroup}(c_1,c_2))$.
\end{proof}

\begin{theorem} For all $X,Y:$ $(X\otimes Y)\,\diff{\otimes}\,Y=X$.
\end{theorem}
\begin{proof} We will prove this for $\Downarrow_{\diff{\oplus}}$ only:
\begin{align*}
&(X\gb_\oplus Y)\Downarrow_{\diff{\oplus}}Y\\
&=\{\,(k,z\diff{\oplus}y)\,|\,(k,z)\in(X\gb_\oplus Y),\,(k,y)\in Y,\,z\not=y\,\}\\
&\skiptext{$=\;$}\uplus\;\{\,(k,z)\,|\,(k,z)\in(X\gb_\oplus Y),\,k\not\in\pi_1(Y)\,\}\\
&=\{\,(k,(x\oplus y)\diff{\oplus}y)\,|\,(k,x)\in X,\,(k,y)\in Y,\\
&\skiptext{$=\{\,(k,(x\oplus y)\diff{\oplus}y)\,|\,\;$}(k,y)\in Y,\,(x\oplus y)\not=y\,\}\\
&\skiptext{$=\;$}\uplus\;\{\,(k,z)\,|\,(k,z)\in(\{\,(k,x\oplus y)\,|\,(k,x)\in X,\,(k,y)\in Y\,\}\\
&\skiptext{$=\;\uplus\;\{\,(k,z)\,|\,(k,z)\in($}\uplus\;\{\,(k,x)\,|\,(k,x)\in X,\,k\not\in\pi_1(Y)\,\}\\
&\skiptext{$=\;\uplus\;\{\,(k,z)\,|\,(k,z)\in($}\uplus\;\{\,(k,y)\,|\,(k,y)\in Y,\,k\not\in\pi_1(X)\,\}),\\
&\skiptext{$=\;\uplus\;\{\,(k,z)\,|\,$}k\not\in\pi_1(Y)\,\}\\
&=\{\,(k,(x\oplus y)\diff{\oplus}y)\,|\,(k,x)\in X,\,(k,y)\in Y,\\
&\skiptext{$=\{\,(k,(x\oplus y)\diff{\oplus}y)\,|\,\;$}(k,y)\in Y,\,(x\oplus y)\not=y\,\}\\
&\skiptext{$=\;$}\uplus\;\{\,(k,x)\,|\,(k,x)\in X,\,k\not\in\pi_1(Y),\,k\not\in\pi_1(Y)\,\}\\
&=\{\,(k,x)\,|\,(k,x)\in X,\,(k,y)\in Y\,\}\\
&\skiptext{$=\;$}\uplus\;\{\,(k,x)\,|\,(k,x)\in X,\,k\not\in\pi_1(Y)\,\}\\
&=X&&\qedhere
\end{align*}
\end{proof}

\end{appendix}
}

\end{document}